\documentclass[12pt,a4paper]{article}
\usepackage[T1]{fontenc}
\usepackage[utf8]{inputenc}
\usepackage[english]{babel}
\usepackage{lmodern,microtype}
\usepackage[a4paper,margin=1in]{geometry}
\usepackage{setspace}
\usepackage{amsmath,amssymb,amsfonts,amsthm,bm}
\usepackage{graphicx,booktabs,tabularx,threeparttable,caption,float}
\usepackage[round,authoryear]{natbib}
\usepackage[hidelinks]{hyperref}
\usepackage{url}
\usepackage{enumitem}
\usepackage{titlesec}
\titleformat{\section}{\large\bfseries}{\thesection.}{0.6em}{}
\titleformat{\subsection}{\normalsize\bfseries}{\thesubsection.}{0.6em}{}
\newcolumntype{Y}{>{\raggedright\arraybackslash}X}
\newcommand{\Normal}{\mathcal{N}}
\newcommand{\invlogit}{\operatorname{logit}^{-1}}
\newcommand{\For}{\mathrm{for}}
\newcommand{\Against}{\mathrm{against}}
\newcommand{\Nonvote}{\mathrm{nonvote}}
\newtheorem{proposition}{Proposition}

\title{\textbf{Public Signals, Concealed Choices:}\\
\large Dynamic Measurement without Behavioral Identification}
\author{Rok Spruk\thanks{Corresponding author: \href{mailto:rok.spruk@ef.uni-lj.si}{rok.spruk@ef.uni-lj.si}}\\University of Ljubljana, School of Economics and Business}
\date{August 2026}

\begin{document}
\maketitle

\begin{abstract}
Members of collective institutions may leave public traces while their individual choices remain concealed. This paper separates a corpus-conditional public position from the behavioral rule linking that position to participation and secret choice. I measure the first with a dynamic ordinal state-space model and establish a likelihood-exclusion result for the second. If behavioral-link parameters enter only the distribution of an entirely unobserved outcome and are a priori independent of measurement parameters, public signals leave their posterior equal to their prior. Declared monotone mappings are therefore reported as sensitivity envelopes, not estimates or identified bounds. Cabinet-sized simulations compare dynamic measurement with a static ordinal model and transparent summaries, exposing gains from temporal pooling and failures under weak measurement or omitted dimensionality. The application reconstructs 41 pre-decision official-source records preceding a concealed Slovenian cabinet decision. Only six of 21 ministers have attributable signals. The model locates a small visible group but leaves the remainder prior-dominated. Event-dependence adjustments, coding perturbations, rolling prediction, and the subsequently disclosed record support the same conclusion. Uncertainty should track the public information environment rather than be converted into unsupported behavioral precision.
\end{abstract}


\section{Introduction}
\label{sec:introduction}

Most models of political positions begin with observed choices. Roll-call models use votes, judicial models use published decisions, and dynamic ideal-point models use repeated behavior to recover movement over time \citep{poole1985spatial,clinton2004statistical,martin2002dynamic}. Text and behavior can also be placed in a common model when observed votes anchor their shared latent structure \citep{kim2018spatial}. Yet many consequential institutions present the reverse information structure. Cabinet ministers, central bankers, regulators, international officials, and corporate directors speak publicly while their individual participation or choices remain confidential. The analyst observes selected public traces, not the behavioral outcome that normally gives an ideal-point model its interpretation.

That reversal creates three distinct inferential problems. First, public communication is not the same construct as private preference or proposal-specific choice. A statement supporting international law may reveal an issue orientation without implying support for one procedural intervention. Second, the public record is selected. Portfolio holders speak more often, parties differ in communications infrastructure, and strategic silence can be informative. Third, even a precisely measured public position does not identify the rule mapping that position into attendance, abstention, or a secret vote. A flexible model cannot solve the third problem by fitting the first more elaborately.

This paper develops a workflow for such settings. Its measurement object is deliberately narrow: the latent directional tendency of attributable communications recoverable under a declared source and coding protocol. I call this a \emph{corpus-conditional public-position state}. It is not a private preference, a complete information set, or an estimate of what an actor would say in an unobserved venue. The distinction treats construct validity and source selection as part of the estimand rather than as preprocessing details \citep{adcock2001measurement,grimmer2013text}. It also prevents no recovered signal from being read as neutrality, private silence, or lack of a position.

The statistical component combines a Gaussian state process with an ordered-logistic observation equation. Dynamic models for coded political text and latent positions are established tools \citep{elff2013dynamic,martin2002dynamic,west1997bayesian}, temporal smoothing can also conceal abrupt change or create misleading precision when the data are thin \citep{reuning2019dynamics}. Accordingly, the contribution is not a new state-space kernel. It is an inferential architecture that fixes the measurement calibration rather than overfitting a cabinet-sized corpus, makes observability limits explicit, discounts repeated communications from the same actor-event, compares the model with simple summaries, and refuses a behavioral translation that the likelihood does not identify.

The formal result supplies the boundary. Let the observed corpus inform latent public positions and measurement parameters, while behavioral-link parameters enter only an entirely unobserved institutional outcome. Under an explicit prior-independence condition, the marginal likelihood is constant in the behavioral parameters, so their posterior equals their prior. This is stronger than saying the vote is difficult to predict. No amount of additional public-signal precision can generate likelihood information about a parameter excluded from the observed-data law. If a correlated prior induces posterior movement, that movement comes from the prior coupling, not from an observed choice. The result applies the general Bayesian logic of nonidentified models to concealed collective choice \citep{poirier1998beliefs}. The paper therefore reports behavioral \emph{sensitivity envelopes}. A declared class of monotone participation-and-choice mappings transforms posterior public-position draws into coherent three-category probability vectors. The resulting ranges show how conclusions vary across assumptions that are external to the signal likelihood. They are not estimated probabilities, confidence intervals, sharp bounds, or a posterior over the calibration class.

Finite-sample experiments evaluate the measurement claim on its own terms. With 21 actors and 16 periods, the dynamic estimator is compared with a static ordinal model, actor means, latest signals, and recency-weighted scores. The dynamic model lowers root mean squared error and produces substantially better interval coverage than the static model in the baseline design, but it does not dominate every simple ranking rule under severe sparsity. Weak discrimination and omitted dimensionality reduce both recovery and coverage. A separate maintained-engine study samples the declared non-centered posterior with four-chain NUTS for 80 simulated datasets and the empirical specifications, accompanied by numerical target-equivalence and cross-engine checks.

The application concerns the Slovenian cabinet's 12 March 2026 decision on intervention in \emph{South Africa v. Israel} before the International Court of Justice. The reconstructed information set contains only material publicly available by 23:59 Central European Time on 11 March 2026. It is retrospective: source discovery and coding occurred after the individual record had been disclosed. Temporal freezing excludes post-decision communications, but it does not turn the exercise into a prospective design. The 41-row restricted official-source corpus contains 36 attributable ministerial signals and five organizational signals. Only six of 21 decision-day ministers have any attributable individual signal, two actors account for 29 of the 36.

The empirical posterior locates a small, highly visible group with a positive public orientation toward external legal and diplomatic pressure. It does not locate the cabinet as a whole. Fifteen ministers remain close to broad prior distributions, and event-level adjustments reduce precision for the most prolific communicators. The subsequently disclosed record of five formal votes for, seven against, and nine members not listed in either voting category, is not used for estimation. Its mismatch with the one-sided public corpus illustrates why issue orientation, formal participation, and proposal-specific choice must remain separate.

The paper makes three contributions. First, it defines a defensible measurement object for selected public traces and embeds that object in a dynamic ordinal model with transparent calibration and dependence diagnostics. Second, it states and proves the likelihood boundary separating public-position measurement from concealed behavior, including the role of prior dependence. Third, it provides an executable evaluation and replication workflow in which uncertainty, negative results, simple benchmarks, coding fragility, and computational diagnostics are part of the evidence. The aim is not to make secrecy disappear. It is to make clear what the public record can support and where inference must stop.

Section~\ref{sec:model} defines the estimand, measurement model, nonidentification result, and calibration envelopes. Section~\ref{sec:simulation} evaluates finite-sample performance. Section~\ref{sec:data} describes the corpus and coding design. Section~\ref{sec:results} reports the empirical measurements and sensitivity analyses. Section~\ref{sec:disclosed} uses the disclosed record as a diagnostic contrast. Section~\ref{sec:conclusion} concludes.

\section{Measurement and Behavioral Ambiguity}
\label{sec:model}

\subsection{The corpus-conditional estimand}

Let $i=1,\ldots,N$ index actors, $t=1,\ldots,T$ calendar periods, and $r=1,\ldots,R$ recovered communications. The source protocol determines which public venues are searched, the temporal cutoff, the attribution rule, and the inclusion criterion. Conditional on that protocol, $\theta_{it}$ denotes the directional tendency of actor $i$'s recoverable public communications at time $t$. Larger values indicate greater publicly expressed support for the defined family of external legal and diplomatic actions. This estimand is conditional in two senses. It conditions on the declared source universe and on a one-dimensional coding calibration. It does not claim that the observation process is ignorable. If the probability that a communication becomes public depends on its content or the actor's latent position, a structural selection model would require additional assumptions or data \citep{rubin1976inference,carrubba2006off,hug2010selection}. The present corpus cannot identify such a model. Instead, the risk set and observed counts are reported, extremity-dependent visibility is examined in simulation, and the substantive interpretation remains limited to the retrieved public record.

\subsection{Dynamic ordinal measurement}

The latent state follows a random walk,
\begin{align}
\theta_{i1} &\sim \Normal(0,\sigma_0^2), \label{eq:initial}\\
\theta_{it} &= \theta_{i,t-1}+u_{it}, \qquad
u_{it}\sim\Normal(0,\tau^2), \quad t=2,\ldots,T. \label{eq:state}
\end{align}
The preferred calibration fixes $\sigma_0=1$ and $\tau=0.30$. Independent actor processes deliberately avoid borrowing a party or portfolio mean for actors without attributable evidence. The choice makes the absence of information visible rather than filling it through group membership. Each recovered individual communication is coded $y_r\in\{-2,-1,0,1,2\}$. With $\widetilde y_r=y_r+3$, the observation equation is
\begin{equation}
\widetilde y_r\mid\theta_{i(r),t(r)}
\sim \operatorname{OrderedLogit}\!\left(
\beta_{m(r)}\theta_{i(r),t(r)},\bm\kappa\right),
\label{eq:ordered}
\end{equation}
where $m(r)$ distinguishes direct from attributed evidence. The cut-points are fixed at $\bm\kappa=(-1.5,-0.5,0.5,1.5)$, with $\beta_{\mathrm{direct}}=1$ and $\beta_{\mathrm{attributed}}=0.75$. These constants orient and scale the posterior and encode the prespecified judgment that attributed evidence is less discriminating. They are calibrations, not corpus-estimated parameters. With only 36 individual records concentrated among six actors and no recovered negative individual category, jointly estimating actor paths, cut-points, source effects, and heterogeneous discrimination would replace transparent assumptions with weakly identified flexibility. The supplement varies state volatility, source inclusion, event dependence, and coding. Organizational statements remain in the published corpus but are excluded from the preferred actor-level likelihood because allocating a government or party statement to particular members would require an additional attribution model. Repeated items within a communications episode are conditionally independent in the item-level fit. Two prespecified sensitivities relax their leverage: event tempering uses a fractional power likelihood giving each actor-event total log-score weight one, and event deduplication replaces each actor-event with a single median-coded observation. Silence is not an ordinal category. For an actor with no recovered item, a posterior centered near zero reflects the zero-centered prior, not evidence of moderation. All position summaries therefore report posterior uncertainty together with actor-level signal counts.

\subsection{A likelihood-exclusion result}

Let $Y$ denote the observed corpus, $\Theta$ the latent public-position trajectories, $\psi$ the measurement parameters and fixed calibration, $B$ the unobserved institutional behavior, and $\lambda$ the parameters mapping a decision-time state into participation and choice. The following proposition states the conditions under which the public corpus cannot update $\lambda$.

\begin{proposition}[Behavioral non-identification]
\label{prop:nonid}
Suppose (i) the observed-data law is $p(Y\mid\Theta,\psi)$ and contains no $\lambda$, (ii) $B$ is entirely unobserved and $\sum_b p(B=b\mid\Theta,\lambda)=1$, and (iii) the prior factorizes as $p(\lambda,\Theta,\psi)=p(\lambda)p(\Theta,\psi)$. Then
\begin{equation}
p(\lambda\mid Y)=p(\lambda).
\label{eq:posterior_prior}
\end{equation}
Thus the public-signal corpus supplies no likelihood information about participation or choice-link parameters, even if it measures $\Theta$ arbitrarily precisely.
\end{proposition}

\begin{proof}
After integrating the unobserved behavior,
\begin{align*}
p(\lambda\mid Y)
&\propto p(\lambda)
\int\!\!\int p(Y\mid\Theta,\psi)p(\Theta,\psi)
\underbrace{\sum_b p(B=b\mid\Theta,\lambda)}_{=1}
\,d\Theta\,d\psi \\
&=p(\lambda)m(Y),
\end{align*}
where $m(Y)$ is constant in $\lambda$. Normalization gives Equation~\eqref{eq:posterior_prior}.
\end{proof}

The prior-independence condition matters. If $p(\lambda,\Theta,\psi)$ is correlated, observing $Y$ can change the marginal distribution of $\lambda$ by updating $\Theta$ or $\psi$. That change is induced by the stipulated prior dependence, it is not evidence from an observed participation or choice. Likewise, appending one disclosed decision after model development does not create a transportable behavioral model. Repeated development decisions with observed individual outcomes, externally validated behavioral information, or defensible exclusion restrictions would be needed to estimate such a mapping.

\subsection{Coherent sensitivity envelopes}

To display behavioral implications without relabeling assumptions as estimates, let $q_i(a,k)=\invlogit(a+k\theta_i^\star)$ be support conditional on formal participation, where $\theta_i^\star=\theta_{iT}$. Let $\pi$ be a common participation probability and declare
\begin{equation}
\mathcal C=\{(a,k,\pi):a\in[-0.5,0.5],\ k\in[0.5,2],\ \pi\in[0.5,0.9]\}.
\label{eq:class}
\end{equation}
Every $c=(a,k,\pi)\in\mathcal C$ implies the coherent vector
\begin{align}
P_i(\For\mid c)&=\pi q_i(a,k), \nonumber\\
P_i(\Against\mid c)&=\pi\{1-q_i(a,k)\}, \label{eq:coherent_vector}\\
P_i(\Nonvote\mid c)&=1-\pi. \nonumber
\end{align}
The reported 95\% posterior sensitivity envelope is the union of equal-tailed 95\% posterior intervals across $c\in\mathcal C$. It propagates measurement uncertainty and shows the consequences of a declared class, but it places no probability distribution over that class. It is therefore neither a Bayesian posterior for behavior nor an identified set. Changing $\mathcal C$ changes the envelope by construction.

\section{Finite-Sample Evaluation}
\label{sec:simulation}

\subsection{Design and comparison estimators}

Each simulation contains $N=21$ actors and $T=16$ periods. States are generated from Equations~\eqref{eq:initial} and \eqref{eq:state}, in the baseline, each actor-period produces one ordinal signal with probability 0.20 and discrimination one. The six stress conditions impose severe sparsity, denser information, extremity-dependent visibility, weaker discrimination, staggered entry, or an omitted second latent dimension. Each condition contains 100 independent replications with recorded seeds.

The target is the decision-time state $\theta_{iT}$. The dynamic ordinal estimator is compared with a static ordinal model, the actor mean, the latest signal, and an exponentially recency-weighted score. Performance is evaluated by root mean squared error, Spearman rank correlation, and, for the two probabilistic models, coverage and width of nominal 95\% intervals, every displayed simulation mean is accompanied by its Monte Carlo standard error. The broad experiment uses a Laplace approximation to the non-centered posterior, checked against elliptical-slice sampling. A separate NUTS experiment fits the declared posterior to 20 datasets in each of four focal conditions.

\begin{table}[!htbp]
\centering
\caption{Finite-Sample Performance of the Dynamic Ordinal Estimator}
\label{tab:simulation}
\scriptsize
\begin{threeparttable}
\setlength{\tabcolsep}{3pt}
\begin{tabular}{lrrrrrr}
\toprule
Condition & Reps. & Signals & RMSE & Rank & Coverage & Width \\
\midrule
Baseline & 100 & 67.1 (0.7) & 1.141 (0.021) & 0.619 (0.017) & 0.945 (0.005) & 4.331 (0.014) \tabularnewline
Severe Sparsity & 100 & 20.4 (0.4) & 1.354 (0.024) & 0.402 (0.021) & 0.944 (0.006) & 5.314 (0.014) \tabularnewline
Dense & 100 & 135.2 (0.9) & 0.928 (0.016) & 0.755 (0.009) & 0.950 (0.005) & 3.634 (0.010) \tabularnewline
Extremity-Dependent Visibility & 100 & 119.8 (1.2) & 0.918 (0.016) & 0.763 (0.011) & 0.951 (0.005) & 3.790 (0.012) \tabularnewline
Weak Discrimination & 100 & 67.2 (0.7) & 1.308 (0.020) & 0.431 (0.020) & 0.891 (0.007) & 4.307 (0.014) \tabularnewline
Turnover & 100 & 56.7 (0.7) & 1.174 (0.022) & 0.618 (0.014) & 0.945 (0.005) & 4.438 (0.014) \tabularnewline
Multidimensional & 100 & 68.3 (0.7) & 1.312 (0.022) & 0.523 (0.017) & 0.893 (0.007) & 4.352 (0.014) \tabularnewline
\bottomrule
\end{tabular}
\begin{tablenotes}[flushleft]\footnotesize
\item \textit{Note:} Except for the replication count, entries are means with Monte Carlo standard errors in parentheses. RMSE and width are in latent-state units, rank is decision-time Spearman correlation, coverage refers to equal-tailed 95\% posterior intervals. Extremity-dependent visibility increases observation probabilities for actors far from the center, it does not estimate or correct a selection mechanism.
\end{tablenotes}
\end{threeparttable}
\end{table}

\begin{figure}[!htbp]
\centering
\includegraphics[width=\textwidth]{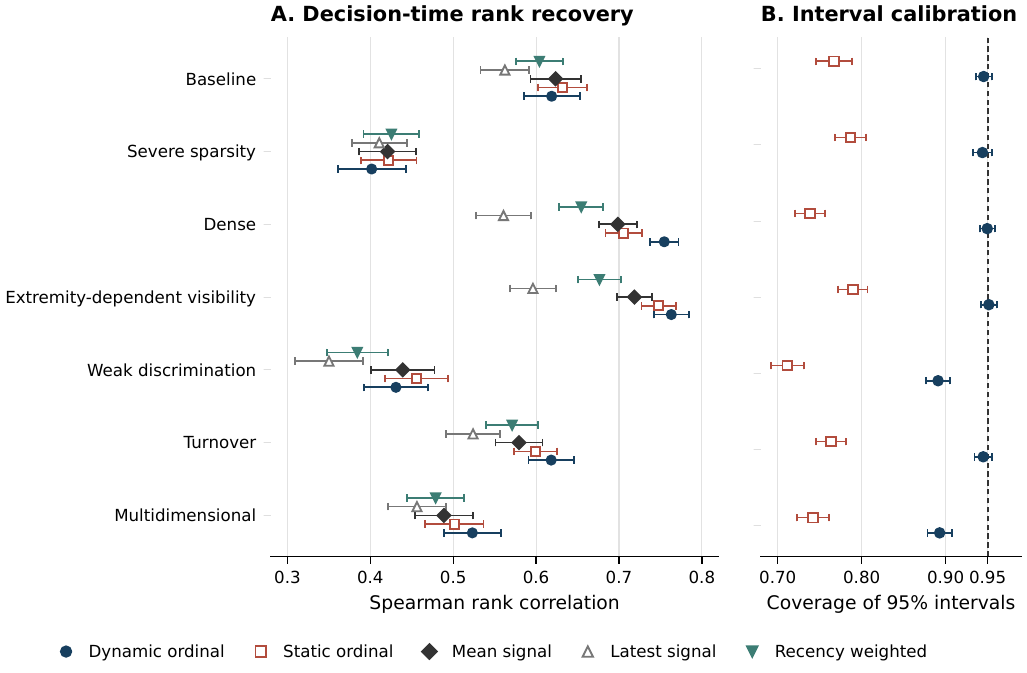}
\caption{Simulation performance across information conditions. Panel A reports mean decision-time Spearman rank correlations for the dynamic ordinal estimator, the static ordinal model, and three transparent signal summaries. Panel B reports equal-tailed 95\% interval coverage for the two probabilistic models, the dashed line marks nominal coverage. Points are condition means and horizontal whiskers extend plus or minus 1.96 Monte Carlo standard errors across 100 independently seeded replications. Each replication contains 21 actors and 16 periods.}
\label{fig:simulation}
\end{figure}

Figure~\ref{fig:simulation} makes the performance tradeoff visible. The baseline rank correlation of the dynamic model is 0.619 (MCSE 0.017), compared with 0.632 (MCSE 0.015) for the static ordinal model and 0.623 (MCSE 0.015) for the actor mean. The dynamic advantage lies instead in lower RMSE, 1.141 (MCSE 0.021) versus 1.212 (MCSE 0.023), and better interval coverage, 0.945 (MCSE 0.005) versus 0.767 (MCSE 0.011). Under dense information, the dynamic model also has the strongest mean ranking correlation, 0.755 (MCSE 0.009), compared with 0.706 (MCSE 0.011) for the static model. Under severe sparsity, several simple methods slightly exceed its rank recovery. Temporal pooling is not a substitute for information.

The visibility condition performs well because its design exposes more extreme states. It therefore shows that the consequences of selective observability depend on the selection rule, it does not show that selection is ignorable. Under weak discrimination and omitted dimensionality, coverage falls to 0.891 (MCSE 0.007) and 0.893 (MCSE 0.007). The estimator should not be described as calibrated when its measurement model is materially wrong.

\subsection{Posterior computation}

The maintained-engine study samples the declared non-centered posterior in PyMC for 20 datasets under each of the baseline, severe-sparsity, weak-discrimination, and multidimensional conditions. Every fit uses four chains, 1,000 warm-up iterations, and 1,000 retained draws. Across 320 chains there are no divergences or maximum-tree-depth hits, the worst rank-normalized $\widehat R$ is 1.012, the minimum bulk and tail effective sample sizes are 3,964 and 1,730, and the minimum E-BFMI is 0.879 \citep{abrilpla2023pymc,hoffman2014nuts,vehtari2021rank}. A 100-state test matches an independent log-density implementation to a maximum absolute difference of $5.7\times10^{-13}$. A mathematically equivalent Stan program and independent JAX NUTS and elliptical-slice implementations provide auditable cross-engine references. These checks establish computational fidelity to the declared posterior, not substantive validity of its assumptions.

\section{Restricted Official-Source Corpus}
\label{sec:data}

\subsection{Temporal and source protocol}

The case concerns the Slovenian cabinet meeting of 12 March 2026. The information cutoff is 23:59 Central European Time on 11 March 2026. Version 1.0 searches publicly accessible GOV.SI and Levica pages using minister names and issue terms covering Israel, Gaza, Palestine, the ICJ, international law, recognition, sanctions, arms restrictions, ceasefire, and humanitarian law. Each retained row records publication date, title, URL, source, mode, actor-event identifier, ordinal code, summary, and coding rationale.

The source universe is intentionally restricted and reproducible, not exhaustive. It omits newspaper archives, television appearances, parliamentary exchanges, deleted social-media posts, and all private communication. It also gives actors with portfolio or party communications responsibilities more opportunities to appear. The estimand and conclusions are conditional on those boundaries. A broader source universe could alter both the recovered set and the actor ordering.

The reconstruction was conducted after disclosure of the individual record. The pre-decision cutoff prevents direct temporal leakage, but source discovery could still reflect hindsight. For that reason, the application is illustrative rather than a clean prospective validation exercise. The disclosed outcomes are stored separately and never enter the measurement likelihood or the behavioral calibration class.

\subsection{Coding and dependence}

The five categories measure a broad \emph{public intervention orientation}: support for external legal, diplomatic, or coercive pressure concerning Israel and Palestine. A code of 2 requires explicit endorsement or implementation of ICJ authority, sanctions, embargoes, diplomatic isolation, or comparably forceful action, 1 denotes qualified support for recognition, international-law enforcement, coordinated pressure, ceasefire, or humanitarian restraint, 0 is relevant but directionally indeterminate, negative values are reserved for qualified or explicit opposition. The scale does not claim that recognition, a ceasefire, sanctions, and intervention in one legal proceeding are behaviorally equivalent. Their compression into one ordinal dimension is a declared measurement choice, tested through a multidimensional simulation stress condition and interpreted accordingly.

The 41-row corpus contains 32 direct individual items, four attributed individual items, and five organizational items. The preferred fit uses the 36 individual items. Tanja Fajon contributes 19, Robert Golob 10, Asta Vrecko four, and Luka Mesec, Simon Maljevac, and Matjaz Han one each. Fifteen ministers have no recovered attributable item. The individual codes comprise one zero, 23 ones, 12 twos, and no negative observation. The data therefore measure degrees within a one-sided retrieved record, they do not empirically span cabinet support and opposition. 

A deterministic, outcome-blind rule applies a prespecified lexical rubric to randomized titles and summaries. It agrees exactly on 35 of 36 individual records, remains within one category on all 36, and has quadratic weighted $\kappa=0.947$. This is a mechanical consistency check, not independent intercoder reliability, because it operates on author-produced fields rather than an independent reading of the original pages. The archive separately provides a randomized, source-linked packet and codebook for a qualified non-author coder, following the distinction between transparent coding procedures and genuine independent reliability evidence \citep{mikhaylov2012coder}. Until that return exists, the paper makes no human intercoder-reliability claim. A one-category perturbation analysis asks whether local miscoding changes the actor ordering without pretending to estimate coder error. The 36 items form 18 actor-events. In the event-tempered specification, each source item receives weight equal to the inverse of its event's item count. In the event-deduplicated specification, each actor-event contributes one median-coded record. These two strategies address a central threat in the corpus, namely, repeated press releases should not mechanically turn communications activity into independent measurement precision.

\section{Empirical Measurement and Sensitivity}
\label{sec:results}

\subsection{Decision-time public positions}

The preferred empirical posterior is fitted with four NUTS chains, each using 1,000 warm-up and 1,000 retained draws. The maximum rank-normalized $\widehat R$ is 1.006, the minimum bulk and tail effective sample sizes are 4,206 and 2,147, there are no divergences or maximum-tree-depth hits, and the minimum E-BFMI is 0.950.

\begin{figure}[!htbp]
\centering
\includegraphics[width=\textwidth]{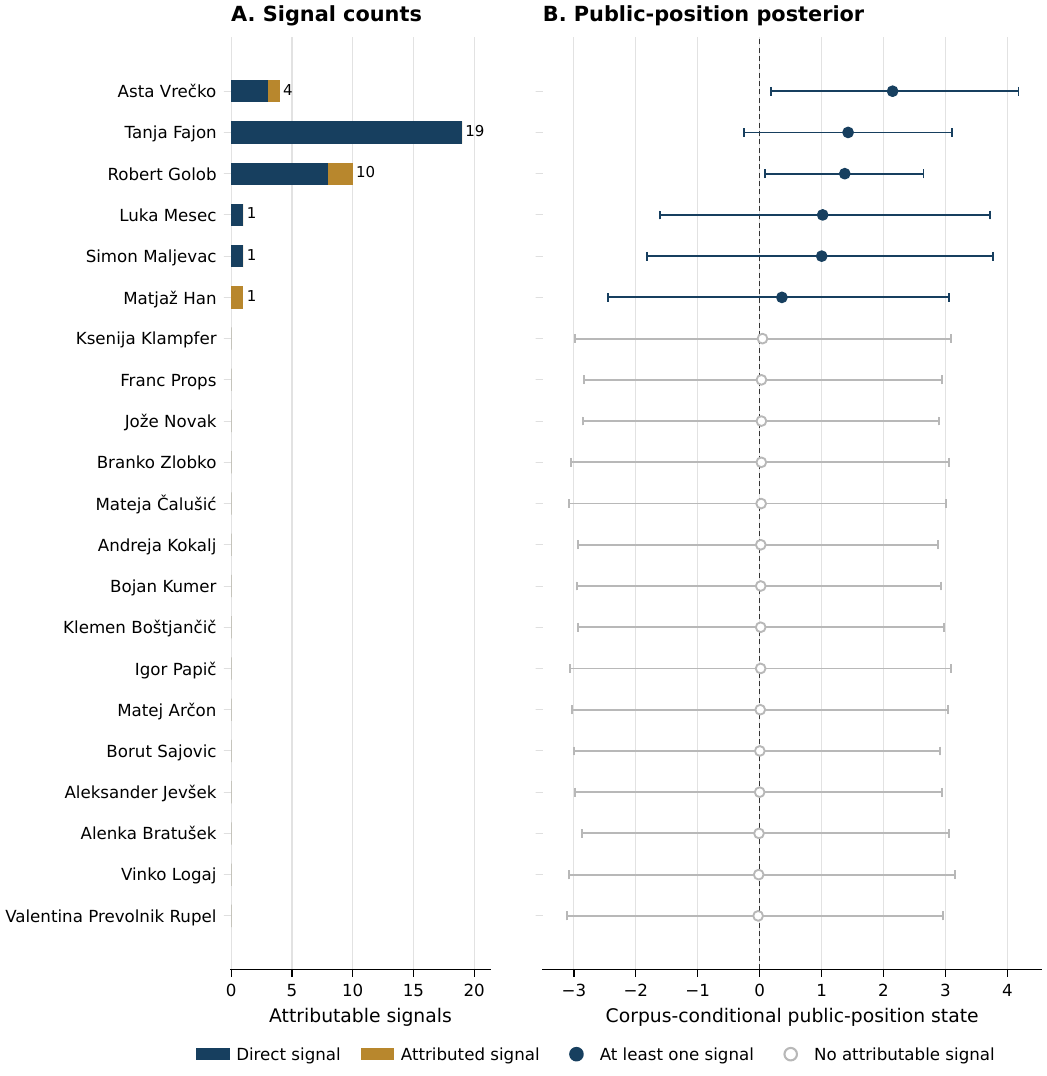}
\caption{Evidence concentration and decision-time public-position states. Panel A reports exact counts of direct and attributed signals in the 36-record actor-level corpus. Panel B reports posterior means and equal-tailed 95\% credible intervals from the preferred item-level NUTS fit. Filled markers denote ministers with at least one attributable signal, hollow gray markers denote ministers whose decision-time posterior is the state-process prior. The retrieval window ends on 11 March 2026, organizational records are excluded.}
\label{fig:positions}
\end{figure}

Figure~\ref{fig:positions} pairs the highly concentrated observation process with the resulting posterior uncertainty. Asta Vrecko has the largest posterior mean, +2.15, with a 95\% credible interval of $[+0.19,+4.18]$. Tanja Fajon follows at +1.43 $[-0.26,+3.10]$ and Robert Golob at 1.37 $[+0.08,+2.65]$. Luka Mesec and Simon Maljevac have means near one but intervals spanning both sides of zero because each contributes one item. Matjaz Han's attributed item yields a mean of +0.36 with an interval of $[-2.45,+3.05]$. The remaining fifteen ministers have means close to zero and intervals of roughly six latent units. Those intervals are the empirical result for unobserved actors and the restricted corpus does not locate them. Absolute latent values are conditional on the fixed scale. The robust statement is narrower. Among recovered communications, Vrecko, Fajon, and Golob form the most clearly positive and best-informed group, while the corpus contains too little actor-specific information to rank most of the cabinet.

\subsection{Dependence, coding, and predictive checks}

Event adjustment preserves the visible ordering but weakens precision. Golob's posterior mean changes from +1.37 in the item-level fit to+ 1.16 under event tempering and +1.25 after deduplication. Fajon moves from +1.43 to +1.30 and +1.40 while Vrecko moves from +2.15 to +2.06 and +2.25. Under event tempering, even Vrecko's 95\% interval slightly crosses zero. This is the intended response to clustered evidence, not instability to be concealed.

The deterministic recoding and one-category perturbation exercise preserve the leading group. Direct-only estimation principally removes Han's weakly informed attributed signal. State volatility preserves the ordering but not the magnitude: raising $\tau$ from 0.15 to 0.50 moves Vrečko's filtered mean from +1.69 to +2.79 and the two single-item estimates from +0.77 to +1.43. This dependence on calibration is substantive and is reported rather than averaged away. A 12-target rolling-origin exercise predicts future actor-event signals using only earlier events. The dynamic model has a lower mean negative log score than the static ordinal model, +1.290 (across-target SE 0.110) versus +1.348 (SE 0.079), but the mean-signal and recency-weighted rules score +1.147 (SE 0.112) and +1.141 (SE 0.114). The targets are dependent within three actors, so the differences are descriptive and support no dominance claim.

\subsection{Behavioral sensitivity envelopes}

\begin{figure}[!htbp]
\centering
\includegraphics[width=0.96\textwidth]{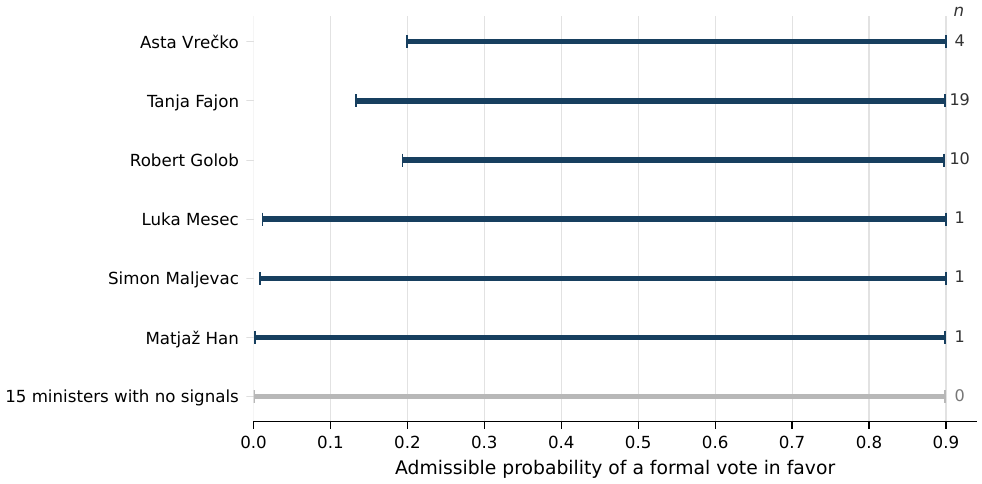}
\caption{Behavioral sensitivity envelopes for a formal vote in favor. For each minister with an attributable signal, the segment is the union of equal-tailed 95\% posterior intervals over the continuous mapping class $\mathcal C$ in Equation~\eqref{eq:class}. The final row reports the common envelope for the 15 ministers without attributable signals. No midpoint is plotted because no probability measure is assigned to $\mathcal C$. These are sensitivity summaries, not estimated choice probabilities or sharp identified bounds. The right column reports attributable-signal counts.}
\label{fig:envelopes}
\end{figure}

Figure~\ref{fig:envelopes} shows that the sensitivity envelopes remain wide even for the visible ministers and nearly cover the admissible probability range for actors without signals. Non-voting probability ranges from 0.10 to 0.50 for every actor because participation is calibrated rather than estimated. The actor-level interval endpoints are reported in the Online Supplement and preserved in the machine-readable replication output. Their construction makes the source of any behavioral precision explicit: relative movement reflects the measured public-position posterior, whereas absolute levels and slopes come from $\mathcal C$. Proposition~\ref{prop:nonid} rules out describing the latter as learned from the corpus.

\section{The Disclosed Record as a Diagnostic Contrast}
\label{sec:disclosed}

The subsequently disclosed official record lists five formal votes in favor and seven against, nine eligible government members appear in neither formal category \citep{gsv2026vote}. The record does not distinguish abstention from nonattendance, so the third category is labeled non-voting. The Prime Minister is not listed as a formal voter but is recorded as having expressed support at the start of the discussion. These outcomes do not estimate the measurement model, select its calibration, or tune $\mathcal C$. The corpus and codebook were reconstructed after disclosure, so presenting the comparison as prospective validation would be misleading. Its role is diagnostic. A retrieved record containing no negative individual statement coexists with seven formal votes against. Several actors without attributable public signals voted against, while several actors with positive public-position estimates did not cast a recorded formal vote.

The divergence is compatible with multiple mechanisms: proposal-specific legal concerns, coalition strategy, attendance, abstention, portfolio responsibility, or a distinction between broad issue orientation and the precise procedural act. The public corpus cannot distinguish among them. That is exactly the point. The disclosed record supplies a sharp empirical warning against treating public rhetoric as a secretly observed vote, but one decision is not enough to estimate a general behavioral link or validate it out of sample.

\begin{table}[!htbp]
\centering
\caption{Public-Signal Availability and Subsequently Disclosed Formal Outcomes}
\label{tab:diagnostic}
\begin{threeparttable}
\begin{tabular}{lrrrr}
\toprule
Pre-decision attributable signals & For & Against & Non-voting & Total \\
\midrule
At least one signal & 3 & 0 & 3 & 6 \\
No attributable signal & 2 & 7 & 6 & 15 \\
\midrule
Total & 5 & 7 & 9 & 21 \\
\bottomrule
\end{tabular}
\begin{tablenotes}[flushleft]\footnotesize
\item \textit{Note:} Entries are descriptive counts. Signal availability is defined using the restricted corpus's pre-decision publication cutoff, although source discovery and coding occurred after disclosure of the outcome. ``Non-voting'' means that the minister was listed in neither formal voting category and does not distinguish abstention from nonattendance. The disclosed outcomes were excluded from every measurement fit. This is a retrospective diagnostic cross-classification, not a forecast-validation test or an estimate of the behavioral link.
\end{tablenotes}
\end{threeparttable}
\end{table}

Table~\ref{tab:diagnostic} provides the relevant comparison without converting public-position estimates into vote predictions. All seven recorded opposing votes occurred among ministers for whom the restricted corpus recovered no attributable signal. Yet absence of a signal is not evidence of opposition. Two ministers in that group voted in favor and six were not listed in either formal voting category. Conversely, three of the six ministers with attributable signals did not cast a recorded formal vote. The contrast therefore reinforces the separation among public observability, issue positioning, formal participation, and proposal-specific choice. These outcomes do not estimate the measurement model, select its calibration, or tune $\mathcal C$. The corpus and codebook were reconstructed after disclosure, so presenting the comparison as prospective validation would be misleading. The analysis accordingly reports no threshold, classification metric, or fitted behavioral equation against the realized outcome. The divergence is compatible with multiple mechanisms such as proposal-specific legal concerns, coalition strategy, attendance, abstention, portfolio responsibility, or a distinction between broad issue orientation and the precise procedural act. The public corpus cannot distinguish among them. That is exactly the point. The disclosed record supplies a sharp empirical warning against treating public rhetoric as a secretly observed vote, but one decision is not enough to estimate a general behavioral link or validate it out of sample.

\section{Conclusion}
\label{sec:conclusion}

Many consequential decisions are made by collective institutions that disclose an aggregate outcome while concealing the participation or choices of individual members. Before those decisions, however, members may leave public traces through speeches, interviews, official statements, press releases, or attributed remarks. This creates a deceptively data-rich inferential setting where public signals are observed, but the behavior of substantive interest is not. The signals are also selected and unevenly distributed. Some actors communicate repeatedly, others remain silent, and silence has no common behavioral meaning. Most importantly, a broad public position on an issue is not equivalent to participation in a particular decision or to a proposal-specific choice. The central problem is therefore not simply how to extract more information from public statements. It is how to determine exactly what those statements identify, and where inference must stop.

Existing methods provide powerful tools for recovering latent positions from observed indicators and for modeling choices when behavioral outcomes or identifying restrictions are available. The unresolved case lies between these two settings. Researchers observe selected indicators of public positioning but not the individual behavior to which they ultimately wish to relate those indicators. In this intermediate setting, precision about a latent public position can easily be relabeled as precision about an unobserved action. This paper fills that methodological gap by separating the measurement of public positions from the behavioral rule linking those positions to participation and concealed choice. It defines a corpus-conditional public-position estimand, measures that object with a dynamic ordinal state-space model, and establishes a likelihood-exclusion result for the behavioral mapping. When behavioral-link parameters do not enter the observed-data law and are a priori independent of the measurement parameters, the public-signal corpus leaves their posterior equal to their prior. Behavioral sensitivity envelopes can therefore show what follows from a declared mapping class, but they cannot transform that class into an empirically learned relationship.

The significance of this contribution lies not in claiming that a dynamic model is universally superior, but in providing an inferential architecture in which the estimand, estimator, diagnostics, and substantive claims remain aligned. The simulations show when temporal pooling is valuable and when it is not. In the baseline design, the dynamic ordinal model reduces latent-state error and produces substantially better interval coverage than a static model. Under dense information, it also delivers the strongest rank recovery. Under severe sparsity, however, transparent summaries remain competitive, while weak discrimination and omitted dimensionality erode calibration. These failures are not peripheral qualifications. They define the method's domain. Temporal structure can organize information that exists, but it cannot manufacture missing signals, eliminate selected observability, or repair a misspecified measurement space.

The Slovenian application demonstrates the empirical importance of these distinctions. Thirty-six attributable records locate only six of 21 ministers, and two ministers account for 29 of those records. Event-level adjustments appropriately weaken the precision generated by repeated communications, while the fifteen ministers without attributable signals retain broad, prior-dominated posteriors. The subsequently disclosed record provides a diagnostic contrast rather than a prospective validation test. All seven recorded opposing votes came from ministers for whom the restricted corpus recovered no attributable signal. Yet silence was not opposition: within that same group, two ministers voted in favor and six were not listed in either formal voting category. Conversely, three of the six ministers with attributable signals did not cast a recorded formal vote. The comparison therefore does not validate a vote-forecasting rule. It demonstrates why public observability, broad issue positioning, formal participation, and proposal-specific choice must be treated as distinct empirical objects.

The external validity of the paper consequently lies in its inferential logic, not in transporting Slovenian actor rankings or numerical estimates to other institutions. Proposition~\ref{prop:nonid} applies wherever the observed likelihood contains public signals but excludes the concealed participation or choice that researchers ultimately wish to explain. That structure arises in central banks, regulatory commissions, judicial conferences, international organizations, coalition governments, corporate boards, and other institutions in which members communicate publicly but decide privately. The measurement model is most useful where signals can be assigned to actors, ordered over time, and coded on a defensible common construct. Its limitations are equally portable. When the source universe is selective, attribution is incomplete, signals are multidimensional, or behavior remains entirely unobserved, posterior precision about public positions must not be interpreted as behavioral identification.

The framework also clarifies what additional evidence would be required to move beyond sensitivity analysis. Prospectively frozen corpora would reduce hindsight in source discovery. Independent human coding would strengthen the measurement layer. Repeated institutional decisions with observed individual outcomes could support estimation of participation and choice mappings, while externally justified exclusion restrictions or validated behavioral information could provide additional leverage. Those extensions would not overturn the paper's identification result but would change the observed-data law so that behavioral parameters could, in principle, be learned.

The practical implication is a disciplined sequence of defining the public measurement object, freezing and document the information set, preserving uncertainty and actor-event dependence, benchmarking the model against transparent alternatives, distinguishing computational verification from substantive validation, stating the assumptions connecting position to behavior, and declining conclusions that the observed likelihood cannot support. A credible analysis should become less certain when public information thins, dependence increases, or the measurement structure weakens. That is not a failure of the model. It is the appropriate inferential response to the evidence. The purpose of a measurement model is not to make secrecy disappear. It is to recover the structure contained in the public record, expose the assumptions required to go further, and identify the point at which evidence gives way to conjecture. In settings structured by secrecy, knowing where to stop is not methodological modesty. It is identification.

\section*{Funding}

This work received no external funding.

\section*{Data Availability Statement}

The replication archive supplied with the manuscript contains the 41-row restricted official-source corpus, actor and actor-event files, source-linked coding materials, the disclosed outcomes in a separate file, all Python and Stan implementations, complete simulation and empirical outputs, machine-readable diagnostics, tables, and compilation instructions. Private communications are neither observed nor counted and are outside the estimand. Upon conditional acceptance, the verified archive will be deposited in the online repository and the permanent citation will be inserted here.

\section*{Competing Interests}

The author declares none.

\bibliographystyle{chicago}
\bibliography{references_new}

\begin{thebibliography}{}

\bibitem[\protect\citeauthoryear{Abril-Pla, Andreani, Carroll, Dong,
  Fonnesbeck, Kochurov, Kumar, Lao, Luhmann, Martin, Osthege, Vieira, Wiecki,
  and Zinkov}{Abril-Pla et~al.}{2023}]{abrilpla2023pymc}
Abril-Pla, O., V.~Andreani, C.~Carroll, L.~Dong, C.~J. Fonnesbeck, M.~Kochurov,
  R.~Kumar, J.~Lao, C.~C. Luhmann, O.~A. Martin, M.~Osthege, R.~Vieira,
  T.~Wiecki, and R.~Zinkov (2023).
\newblock {PyMC}: A modern and comprehensive probabilistic programming
  framework in {Python}.
\newblock {\em PeerJ Computer Science\/}~{\em 9}, e1516.

\bibitem[\protect\citeauthoryear{Adcock and Collier}{Adcock and
  Collier}{2001}]{adcock2001measurement}
Adcock, R. and D.~Collier (2001).
\newblock Measurement validity: A shared standard for qualitative and
  quantitative research.
\newblock {\em American Political Science Review\/}~{\em 95\/}(3), 529--546.

\bibitem[\protect\citeauthoryear{Carrubba, Gabel, Murrah, Clough, Montgomery,
  and Schambach}{Carrubba et~al.}{2006}]{carrubba2006off}
Carrubba, C.~J., M.~Gabel, L.~Murrah, R.~Clough, E.~Montgomery, and
  R.~Schambach (2006).
\newblock Off the record: Unrecorded legislative votes, selection bias and
  roll-call vote analysis.
\newblock {\em British Journal of Political Science\/}~{\em 36\/}(4), 691--704.

\bibitem[\protect\citeauthoryear{Clinton, Jackman, and Rivers}{Clinton
  et~al.}{2004}]{clinton2004statistical}
Clinton, J.~D., S.~Jackman, and D.~Rivers (2004).
\newblock The statistical analysis of roll call data.
\newblock {\em American Political Science Review\/}~{\em 98\/}(2), 355--370.

\bibitem[\protect\citeauthoryear{Elff}{Elff}{2013}]{elff2013dynamic}
Elff, M. (2013).
\newblock A dynamic state-space model of coded political texts.
\newblock {\em Political Analysis\/}~{\em 21\/}(2), 217--232.

\bibitem[\protect\citeauthoryear{{General Secretariat of the Government of the
  Republic of Slovenia}}{{General Secretariat of the Government of the Republic
  of Slovenia}}{2026}]{gsv2026vote}
{General Secretariat of the Government of the Republic of Slovenia} (2026,
  March).
\newblock Record of voting on item 6 of the 194th regular session: Proposed
  intervention of the {Republic of Slovenia} in {South Africa v. Israel}.
\newblock Internal cabinet record dated March 12, 2026, subsequently disclosed
  following proceedings before the Information Commissioner.

\bibitem[\protect\citeauthoryear{Grimmer and Stewart}{Grimmer and
  Stewart}{2013}]{grimmer2013text}
Grimmer, J. and B.~M. Stewart (2013).
\newblock Text as data: The promise and pitfalls of automatic content analysis
  methods for political texts.
\newblock {\em Political Analysis\/}~{\em 21\/}(3), 267--297.

\bibitem[\protect\citeauthoryear{Hoffman and Gelman}{Hoffman and
  Gelman}{2014}]{hoffman2014nuts}
Hoffman, M.~D. and A.~Gelman (2014).
\newblock The {No-U-Turn Sampler}: Adaptively setting path lengths in
  {Hamiltonian Monte Carlo}.
\newblock {\em Journal of Machine Learning Research\/}~{\em 15\/}(47),
  1593--1623.

\bibitem[\protect\citeauthoryear{Hug}{Hug}{2010}]{hug2010selection}
Hug, S. (2010).
\newblock Selection effects in roll call votes.
\newblock {\em British Journal of Political Science\/}~{\em 40\/}(1), 225--235.

\bibitem[\protect\citeauthoryear{Kim, Londregan, and Ratkovic}{Kim
  et~al.}{2018}]{kim2018spatial}
Kim, I.~S., J.~Londregan, and M.~Ratkovic (2018).
\newblock Estimating spatial preferences from votes and text.
\newblock {\em Political Analysis\/}~{\em 26\/}(2), 210--229.

\bibitem[\protect\citeauthoryear{Martin and Quinn}{Martin and
  Quinn}{2002}]{martin2002dynamic}
Martin, A.~D. and K.~M. Quinn (2002).
\newblock Dynamic ideal point estimation via markov chain monte carlo for the
  {U.S.} supreme court, 1953--1999.
\newblock {\em Political Analysis\/}~{\em 10\/}(2), 134--153.

\bibitem[\protect\citeauthoryear{Mikhaylov, Laver, and Benoit}{Mikhaylov
  et~al.}{2012}]{mikhaylov2012coder}
Mikhaylov, S., M.~Laver, and K.~R. Benoit (2012).
\newblock Coder reliability and misclassification in the human coding of party
  manifestos.
\newblock {\em Political Analysis\/}~{\em 20\/}(1), 78--91.

\bibitem[\protect\citeauthoryear{Poirier}{Poirier}{1998}]{poirier1998beliefs}
Poirier, D.~J. (1998).
\newblock Revising beliefs in nonidentified models.
\newblock {\em Econometric Theory\/}~{\em 14\/}(4), 483--509.

\bibitem[\protect\citeauthoryear{Poole and Rosenthal}{Poole and
  Rosenthal}{1985}]{poole1985spatial}
Poole, K.~T. and H.~Rosenthal (1985).
\newblock A spatial model for legislative roll call analysis.
\newblock {\em American Journal of Political Science\/}~{\em 29\/}(2),
  357--384.

\bibitem[\protect\citeauthoryear{Reuning, Kenwick, and Fariss}{Reuning
  et~al.}{2019}]{reuning2019dynamics}
Reuning, K., M.~R. Kenwick, and C.~J. Fariss (2019).
\newblock Exploring the dynamics of latent variable models.
\newblock {\em Political Analysis\/}~{\em 27\/}(4), 503--517.

\bibitem[\protect\citeauthoryear{Rubin}{Rubin}{1976}]{rubin1976inference}
Rubin, D.~B. (1976).
\newblock Inference and missing data.
\newblock {\em Biometrika\/}~{\em 63\/}(3), 581--592.

\bibitem[\protect\citeauthoryear{Vehtari, Gelman, Simpson, Carpenter, and
  B{\"u}rkner}{Vehtari et~al.}{2021}]{vehtari2021rank}
Vehtari, A., A.~Gelman, D.~Simpson, B.~Carpenter, and P.-C. B{\"u}rkner (2021).
\newblock Rank-normalization, folding, and localization: An improved
  {$\widehat{R}$} for assessing convergence of {MCMC}.
\newblock {\em Bayesian Analysis\/}~{\em 16\/}(2), 667--718.

\bibitem[\protect\citeauthoryear{West and Harrison}{West and
  Harrison}{1997}]{west1997bayesian}
West, M. and J.~Harrison (1997).
\newblock {\em Bayesian Forecasting and Dynamic Models\/} (2 ed.).
\newblock New York: Springer.

\end{thebibliography}
\end{document}


\maketitle
\tableofcontents
\newpage

\section{Scope and Inferential Workflow}

This supplement is the paper's audit trail. It separates four objects that would otherwise be easy to conflate: a latent public-position state, the public observation and coding process, a proposal-specific behavioral link, and prediction of later public signals. Table~\ref{tabS:claimmap} states what updates each object, how it is reported, and which stronger claim is ruled out.

\begin{table}[!htbp]
\centering
\caption{Inferential Objects, Evidence, and Claim Boundaries}
\label{tabS:claimmap}
\small
\begin{tabularx}{\textwidth}{P{0.18\textwidth}P{0.25\textwidth}P{0.24\textwidth}Y}
\toprule
Object & Evidence & Reported quantity & Excluded claim \\
\midrule
Public-position state, $\theta_{it}$ & 36 dated direct or attributed ordinal records for six of 21 ministers & Calibration-conditional posterior state and uncertainty & Private preference, neutrality of silent actors, or proposal-specific choice \\
\addlinespace[0.25em]
Public observability and coding & Complete 41-row restricted-source inventory, event clusters, rule check, and declared perturbations & Coverage diagnostics and sensitivity analyses & Exhaustive communication census, ignorable selection, or human reliability before an external return \\
\addlinespace[0.25em]
Behavioral link, $\lambda$ & No behavioral outcome enters the measurement likelihood & Prior-invariance result and envelopes over a declared mapping class & Estimated participation, vote probability, or sharp identified set \\
\addlinespace[0.25em]
Future public signal & Twelve rolling-origin actor-event targets, using earlier events only & Descriptive proper scores and their across-target uncertainty & Prospective secret-vote validation or predictive dominance \\
\bottomrule
\end{tabularx}
\begin{minipage}{0.96\textwidth}
\footnotesize
\textit{Note:} The empirical corpus covers official GOV.SI and Levica sources published no later than 11 March 2026. It was reconstructed after disclosure of the voting record. Temporal freezing excludes post-decision material from estimation but does not convert the reconstruction into a prospective design.
\end{minipage}
\end{table}

The workflow follows this separation. Section~S3 defines the roster, source universe, coding rules, and actor-event construction. Section~S4 gives the state process, ordinal likelihood, calibration, and behavioral nonidentification result. Sections~S5 and S6 report implementation checks and simulation evidence. Sections~S7--S9 report the empirical posterior, rolling-origin signal prediction, and sensitivity analyses. Section~S10 uses the subsequently disclosed record only as a qualitative diagnostic contrast. Sections~S11--S13 document reproducibility, verification status, and remaining limits.

Version 1.0 publishes source URLs, titles, paraphrased summaries, coding rationales, actor-event identifiers, and a separate outcome-blind rule-based recoding. Repeated communications on the same actor-specific episode remain visible but are addressed through event-tempered and event-deduplicated sensitivities. A randomized, source-linked packet is supplied for independent human recoding; because that return is not yet available, no human intercoder statistic is reported.

The computational audit has five layers. A 100-replication, seven-condition benchmark compares the dynamic estimator with a static ordinal model and transparent signal summaries. A separate maintained-engine experiment fits four-chain PyMC NUTS models to 20 datasets in each of four focal conditions. The empirical item-level and actor-event specifications use the same NUTS settings. A 100-state log-density test checks the executed target against an independent implementation, and JAX NUTS plus elliptical-slice results provide cross-algorithm comparisons. Rolling-origin prediction then evaluates later public actor-event signals---never the concealed vote.

\section{Methodological Location and Scope}

The framework composes established elements where their separation is substantively consequential. The latent-state component follows dynamic Bayesian measurement \citep{martin2002dynamic,west1997bayesian,elff2013dynamic}; the observation equation follows ordinal latent-response logic \citep{albert1993bayesian}. Joint text--vote models can use observed choices to anchor a shared scale \citep{kim2018spatial}. Here, the target behavior is concealed at the point of estimation, so that anchor is unavailable by design.

The corpus protocol makes construct validity and coding transparency part of the estimand rather than treating them as innocuous preprocessing \citep{adcock2001measurement,grimmer2013text,mikhaylov2012coder}. Observability may depend on the latent state or on actor characteristics \citep{rubin1976inference}, and the availability of recorded political choices can itself be selected \citep{carrubba2006off,hug2010selection}. Evaluation therefore distinguishes state recovery, interval coverage, and out-of-sample signal prediction \citep{shmueli2010explain,cranmer2017predictive,gneiting2007strictly}.

No component is claimed as unprecedented in isolation. The contribution is a disciplined composition for settings with public traces but unavailable behavior: define a narrow public-position estimand, expose the information set, preserve dependence and coding uncertainty, benchmark against simple alternatives, and prove when the measurement likelihood cannot update a behavioral mapping. The resulting envelopes are declared sensitivity summaries. They are not substitutes for observed choices.

\section{Data Files and Temporal Construction}

\subsection{Decision-day roster}

The decision-day roster contains 21 ministers. The actor file records a stable numeric identifier, the minister's name, formal entry date, entry quarter, and broad political grouping. From 2022Q2 through 2026Q1, tenure adjustment produces 247 actor-quarter observations rather than the 336 cells implied by assigning all 21 decision-day ministers to every quarter. Extending the descriptive panel through 2026Q2 would produce 268 actor-quarters, but Q2 2026 is excluded from the frozen fit.

\begin{longtable}{rlll}
\caption{Decision-Day Roster and Entry Timing}\label{tabS:roster}\\
\toprule
ID & Minister & Entry date & Entry quarter \\
\midrule
\endfirsthead
\toprule
ID & Minister & Entry date & Entry quarter \\
\midrule
\endhead
1 & Robert Golob & 1 June 2022 & 2022Q2 \\
2 & Tanja Fajon & 1 June 2022 & 2022Q2 \\
3 & Luka Mesec & 1 June 2022 & 2022Q2 \\
4 & Klemen Boštjančič & 1 June 2022 & 2022Q2 \\
5 & Matej Arčon & 1 June 2022 & 2022Q2 \\
6 & Matjaž Han & 1 June 2022 & 2022Q2 \\
7 & Igor Papič & 1 June 2022 & 2022Q2 \\
8 & Asta Vrečko & 1 June 2022 & 2022Q2 \\
9 & Bojan Kumer & 1 June 2022 & 2022Q2 \\
10 & Aleksander Jevšek & 1 June 2022 & 2022Q2 \\
11 & Alenka Bratušek & 24 January 2023 & 2023Q1 \\
12 & Simon Maljevac & 24 January 2023 & 2023Q1 \\
13 & Valentina Prevolnik Rupel & 13 October 2023 & 2023Q4 \\
14 & Franc Props & 7 December 2023 & 2023Q4 \\
15 & Jože Novak & 7 December 2023 & 2023Q4 \\
16 & Mateja Čalušić & 12 January 2024 & 2024Q1 \\
17 & Borut Sajovic & 7 October 2024 & 2024Q4 \\
18 & Vinko Logaj & 7 October 2024 & 2024Q4 \\
19 & Ksenija Klampfer & 17 December 2024 & 2024Q4 \\
20 & Andreja Kokalj & 21 November 2025 & 2025Q4 \\
21 & Branko Zlobko & 21 November 2025 & 2025Q4 \\
\bottomrule
\end{longtable}

\subsection{Signal-level schema}

The file \texttt{data/signals.csv} contains one row per reconstructed communication or organizational act. Its principal fields are:

\begin{longtable}{P{0.23\textwidth}P{0.69\textwidth}}
\toprule
Field & Definition \\
\midrule
\endfirsthead
\toprule
Field & Definition \\
\midrule
\endhead
\texttt{signal\_id} & Stable row identifier. \\
\texttt{actor\_id} & Minister identifier for direct and attributed signals; blank for organizational signals. \\
\texttt{publication\_date} & Public availability date in ISO format. \\
\texttt{period\_id} & Calendar quarter used in the state-space model. \\
\texttt{source\_type} & GOV.SI or Levica in version 1.0. \\
\texttt{signal\_mode} & Direct, attributed, or organizational. \\
\texttt{ordinal\_code} & Directional category in $\{-2,-1,0,1,2\}$. \\
\texttt{source\_title} & Published title of the webpage. \\
\texttt{source\_url} & Permanent source address. \\
\texttt{signal\_summary} & Paraphrase of the relevant content. \\
\texttt{coding\_rationale} & Reason for the assigned category. \\
\texttt{predecision} & Indicator that the source predates the 12 March 2026 decision. \\
\bottomrule
\end{longtable}

\subsection{Search and inclusion protocol}

The declared source universe is restricted to GOV.SI and Levica. The reconstruction searched combinations of every decision-day minister's name with Israel, Gaza, Palestine, ICJ, international law, recognition, sanctions, arms, ceasefire, and humanitarian law. The analytical freeze admits only material published by 11 March 2026, the day before the decision. A page was retained when its public content could bear on the direction of a minister's position toward international legal or diplomatic action involving Israel. Purely humanitarian description without a defensible policy implication was coded zero or excluded. Collective government or party acts without unique ministerial attribution were retained as organizational records but excluded from the preferred individual likelihood.

The unit is a distinct public page or act, identified by a stable \texttt{signal\_id}. Multiple communications concerning one continuing policy episode are not silently collapsed: their original rows remain in \texttt{signals.csv}, while \texttt{event\_id} records actor-specific episode membership for dependence analyses. The corpus preserves source title, URL, paraphrased content, and coding rationale so that an auditor can reconstruct every inclusion. Inaccessible links encountered during independent recoding must be recorded and repaired or archived; they may not be replaced through an undisclosed search.

This protocol is reproducible but deliberately narrow. It does not exhaust newspaper interviews, parliamentary exchanges, broadcasts, social-media posts, or deleted material. Only six ministers have a recovered individual record. That negative search result describes coverage under the stated official-source protocol; it is neither a missing-at-random assumption nor evidence that other ministers were substantively neutral.

\subsection{Public-signal sampling frame and unobserved private communication}

The empirical sampling frame contains only attributable communications retrievable from the declared public-source universe. Private meetings, cabinet deliberations, informal conversations, direct messages, unpublished memoranda, and undocumented exchanges are outside it. Their number, direction, timing, and distribution are unknown. Consequently, the retrieved item count is a measure of documented public evidence under this protocol, not an estimate of an actor's total communication.

A zero in the actor-period matrix means only that the protocol recovered no attributable public item for that actor and quarter. It does not imply no position, no private communication, no internal participation, or ideological centrism. No private signal is sought, imputed, treated as observed, or used for empirical validation. Simulation conditions vary a stylized public observation process to evaluate estimator behavior; they do not estimate the empirical quantity of private communication.

\subsection{Coding rules}

The scale is tied to the specific intervention dimension:

\begin{table}[!htbp]
\centering
\caption{Ordinal Coding Rules}
\label{tabS:codebook}
\begin{tabularx}{\textwidth}{c l Y}
\toprule
Code & Category & Rule \\
\midrule
$2$ & Explicit support & Endorsement of ICJ authority or participation, sanctions, embargoes, diplomatic isolation, or similarly forceful legal action. \\
$1$ & Qualified support & Support for recognition of Palestine, international-law enforcement, coordinated pressure, ceasefire, or humanitarian restraint without endorsing the exact intervention. \\
$0$ & Balanced or indeterminate & Relevant statement without a defensible directional implication for intervention. \\
$-1$ & Qualified opposition & Procedural caution, restraint, or diplomatic objection weighing against intervention. \\
$-2$ & Explicit opposition & Direct rejection of intervention, ICJ involvement, sanctions, or comparable action. \\
\bottomrule
\end{tabularx}
\end{table}

The reconstructed individual component of the restricted official-source corpus contains one zero, 23 category-one signals, and 12 category-two signals. It contains no recovered negative individual signal. This is a property of the restricted official-source corpus, not a claim about private positions or the absence of cabinet opposition.

The five labels are ordered, not cardinal measurements. The model never treats the numerical gap between adjacent raw codes as an observed equal interval. Instead, it maps the ordering through four declared cut-points in an ordered-logistic likelihood. The one-sided empirical distribution is therefore a genuine information limitation: the data can distinguish degrees of publicly supportive content more readily than they can locate an opposition pole that is never observed in the restricted corpus.

\subsection{Algorithmic consistency check and independent human recoding}

The archive retains an outcome-blind rule-based recoding generated from randomized titles and author-produced summaries under a declared lexical rubric. Among the 36 direct or attributed items, it agrees exactly with the original code on 35, remains within one category on all 36, and yields quadratic weighted $\kappa=0.947$ and linear weighted $\kappa=0.943$. This is an algorithmic consistency diagnostic only. It is not human intercoder reliability, because the rule operates on text fields produced during corpus construction and cannot supply an independent human interpretation of the original public communication.

The preferred reliability exercise instead treats the author and one qualified non-author as two human coders. The external coder receives a randomized 41-item packet containing links to the original public pages, source titles, actor attribution when it appears in the public source, source mode, and blank coding fields. The coder does not receive the original category, rationale, event identifier, model output, manuscript results, disclosed vote, or private crosswalk. Preserving the full public page takes precedence over artificial identity blinding; the design is outcome-blinded and code-blinded, and the coder attests that the disclosed individual outcomes were unknown during coding.

The codebook defines the target as the position publicly expressed in the linked communication, not a private preference, likely attendance, or secret vote. Coders are instructed neither to seek nor infer private meetings, unpublished messages, or undocumented signals. Source-access problems are recorded rather than replaced by outside searches. The coder also records confidence, the passage or page location supporting the category, access date, source language, and any translation used.

The returned packet and attestation are hashed and timestamped before the private crosswalk is opened. The unadjudicated analysis reports exact and within-one-category agreement, mean absolute category difference, linear and quadratic weighted $\kappa$, ordinal Krippendorff's $\alpha$, and actor-event cluster-bootstrap intervals, separately for all 41 records and the 36 direct or attributed records. Disagreements are frozen in a machine-readable file. Any adjudication is a subsequent, separately logged stage and never replaces the unadjudicated reliability statistics. The merge script refuses to report human reliability unless all 41 classifications, source-evidence fields, and the signed independence attestation pass validation.

\paragraph{External-coding status.} The source-linked packet and integrity checks are complete, but no human-reliability statistic is reported until the qualified non-author returns all 41 classifications and signs the independence attestation.%

\subsection{Actor-event construction}

Signals are clustered into actor-specific events using date, quarter, policy episode, and substantive continuity. For example, Fajon's repeated October--December 2023 ceasefire diplomacy forms one event, her January 2024 ICJ statements another, and her May 2024 Rafah statements a third. Golob's recognition-of-Palestine communications and Vrečko's sanctions-related communications are treated analogously. The resulting machine-readable file \path{data/actor_event_signals.csv} contains one row per actor-event.

Two robustness estimators address within-event dependence. The event-tempered analysis raises each item likelihood to the reciprocal of its actor-event size, so every event contributes total log-score weight one. Because fractional powers define a generalized-Bayes sensitivity target rather than a fully specified generative dependence model, its intervals are interpreted only as a robustness comparison. The event-deduplicated specification instead replaces each event with its median ordinal category, rounds half-category ties toward zero, and uses the mean mode discrimination. Both preserve temporal order while preventing repeated statements from mechanically multiplying precision; neither purports to estimate an unknown within-event correlation structure from 18 events.

\subsection{Complete source inventory}

\begingroup
\footnotesize
\setlength{\tabcolsep}{3pt}
\begin{longtable}{@{}r >{\raggedright\arraybackslash}p{0.10\textwidth} >{\raggedright\arraybackslash}p{0.16\textwidth} >{\raggedright\arraybackslash}p{0.13\textwidth} r >{\raggedright\arraybackslash}p{0.40\textwidth}@{}}
\caption{Complete Version 1.0 Signal Inventory}\label{tabS:inventory}\\
\toprule
ID & Date & Actor or organization & Mode & Code & Source title \\
\midrule
\endfirsthead
\toprule
ID & Date & Actor or organization & Mode & Code & Source title \\
\midrule
\endhead
1 & 2023-10-10 & Tanja Fajon & direct & 0 & We must give people in the Middle East hope for peace \tabularnewline
2 & 2023-10-23 & Tanja Fajon & direct & 1 & Minister Fajon calls for a humanitarian ceasefire \tabularnewline
3 & 2023-11-13 & Tanja Fajon & direct & 1 & Crisis of humanity in Gaza \tabularnewline
4 & 2023-11-20 & Tanja Fajon & direct & 1 & Immediate humanitarian ceasefire in Gaza \tabularnewline
5 & 2023-11-22 & Tanja Fajon & direct & 1 & EU Mediterranean countries must raise their voice \tabularnewline
6 & 2023-11-24 & Tanja Fajon & direct & 1 & Peace plan should lead to recognition of a Palestinian state \tabularnewline
7 & 2023-11-27 & Tanja Fajon & direct & 1 & We all want peace \tabularnewline
8 & 2023-11-29 & Tanja Fajon & direct & 1 & A lasting ceasefire in Gaza is urgent \tabularnewline
9 & 2023-12-06 & Tanja Fajon & direct & 1 & Respect international humanitarian law and humanitarian pause \tabularnewline
10 & 2024-01-11 & Tanja Fajon & direct & 2 & Slovenia in the ICJ advisory-opinion proceedings \tabularnewline
11 & 2024-01-17 & Tanja Fajon & direct & 2 & It is high time to end suffering in Gaza \tabularnewline
12 & 2024-05-05 & Tanja Fajon & direct & 1 & Calls for dialogue and for Israel not to attack Rafah \tabularnewline
13 & 2024-05-12 & Tanja Fajon & direct & 1 & Israel and Hamas must resume talks \tabularnewline
14 & 2024-05-14 & Tanja Fajon & direct & 1 & International community must come together \tabularnewline
15 & 2024-05-27 & Tanja Fajon & direct & 2 & Condemnation of the Rafah refugee-camp attack \tabularnewline
16 & 2024-06-05 & Tanja Fajon & direct & 2 & Slovenia and Palestine establish diplomatic relations \tabularnewline
17 & 2024-09-19 & Tanja Fajon & direct & 1 & West Bank at boiling point \tabularnewline
18 & 2024-10-07 & Tanja Fajon & direct & 1 & The key to peace is to address root causes \tabularnewline
19 & 2024-10-14 & Tanja Fajon & direct & 2 & Further steps in crisis areas \tabularnewline
20 & 2023-11-09 & Robert Golob & direct & 1 & A ceasefire is essential for civilian aid \tabularnewline
21 & 2024-04-16 & Robert Golob & attributed & 1 & Responsibility in the international community \tabularnewline
22 & 2024-05-27 & Robert Golob & direct & 1 & Statement on the recognition of Palestine \tabularnewline
23 & 2024-05-30 & Robert Golob & attributed & 1 & Halfway through the term of office \tabularnewline
24 & 2025-05-16 & Robert Golob & direct & 1 & Europe must respond with cohesion and openness \tabularnewline
25 & 2025-05-29 & Robert Golob & direct & 2 & Joint statement by Slovenia and Spain \tabularnewline
26 & 2025-06-04 & Robert Golob & direct & 1 & Anniversary of recognition of Palestine \tabularnewline
27 & 2025-09-22 & Robert Golob & direct & 1 & Recognition of Palestine by additional countries \tabularnewline
28 & 2026-01-19 & Robert Golob & direct & 2 & We want to be the voice of reason \tabularnewline
29 & 2026-02-06 & Robert Golob & direct & 1 & Principled foreign policy opens doors to cooperation \tabularnewline
30 & 2024-05-30 & Asta Vrečko & attributed & 1 & Halfway through the term of office \tabularnewline
31 & 2024-06-04 & Asta Vrečko & direct & 2 & Slovenia recognized Palestine \tabularnewline
32 & 2025-09-25 & Asta Vrečko & direct & 2 & Netanyahu prohibited from entering Slovenia \tabularnewline
33 & 2025-10-04 & Asta Vrečko & direct & 2 & Eighth regular congress of Levica \tabularnewline
34 & 2024-06-04 & Luka Mesec & direct & 2 & Slovenia recognized Palestine \tabularnewline
35 & 2024-06-04 & Simon Maljevac & direct & 2 & Slovenia recognized Palestine \tabularnewline
36 & 2024-05-30 & Matjaž Han & attributed & 1 & Halfway through the term of office \tabularnewline
37 & 2023-10-17 & Government of Slovenia & organizational & 1 & Slovenia increases humanitarian aid for Gaza \tabularnewline
38 & 2024-05-09 & Government of Slovenia & organizational & 1 & 102nd regular session of the Government \tabularnewline
39 & 2025-07-31 & Government of Slovenia & organizational & 2 & Prohibition of weapons trade with Israel \tabularnewline
40 & 2024-10-04 & Levica & organizational & 2 & Weapons for Israel must not pass through Koper \tabularnewline
41 & 2024-11-23 & Levica & organizational & 2 & Stop procurement of Israeli weapons \tabularnewline
\bottomrule
\end{longtable}
\endgroup

The CSV file contains the corresponding URLs, summaries, and coding rationales. The inventory table is not a substitute for the machine-readable corpus. Exact quarterly counts of direct and attributed records are preserved in \path{output/stata/quarterly_signal_counts.csv}; the five organizational records are excluded. Empty quarters mean that the declared search protocol recovered no attributable public item, not neutrality or an absence of private communication.

\section{Statistical Model}

\subsection{Latent public-position state process}

For each actor $i$, the initial public-position state and innovations satisfy
\begin{align}
\theta_{i1}&\sim\Normal(0,\sigma_0^2),\\
\theta_{it}-\theta_{i,t-1}&\sim\Normal(0,\tau^2),\qquad t=2,\ldots,T.
\end{align}
The preferred calibration fixes $\sigma_0=1$, $\tau=0.30$, and $T=16$ quarters from 2022Q2 through 2026Q1. The implied prior standard deviation at period $t$ is $\sqrt{\sigma_0^2+(t-1)\tau^2}$, so uncertainty grows in the absence of observations.

The model uses this common calendar grid for all decision-day actors. Roster entry dates determine the descriptive minister-quarter risk set in Section~S3.1 but do not truncate the latent measurement process: $\theta_{it}$ denotes an actor's potentially measurable public orientation, not an office-tenure outcome. No retained attributable record predates its actor's cabinet entry. For a completely silent actor, however, the common-grid convention implies the same decision-time prior regardless of entry date. This convention is disclosed because a tenure-started prior would narrow intervals for late entrants without adding evidence.

The prior is independent across actors. Hierarchical party pooling would narrow the posterior of silent ministers by transferring information from vocal copartisans. That may be useful with many actors and repeated decisions, but here it would blur the distinction between individual and organizational evidence and would create apparently individualized precision from group assumptions.

\subsection{Ordered-logistic measurement}

Let $\widetilde y_r\in\{1,\ldots,5\}$ index the ordered categories $-2,\ldots,2$, let $\eta_r=\beta_{m(r)}\theta_{i(r),t(r)}$, and set $\kappa_0=-\infty$ and $\kappa_5=\infty$. Then
\begin{align}
\Pr(\widetilde y_r\leq j\mid\theta)&=\invlogit(\kappa_j-\eta_r),\qquad j=1,\ldots,4,\\
p_{rj}(\theta)&=\invlogit(\kappa_j-\eta_r)-\invlogit(\kappa_{j-1}-\eta_r).
\end{align}
The item-level log likelihood is $\sum_r\log p_{r,\widetilde y_r}(\theta)$. Fixed symmetric cut-points $(-1.5,-0.5,0.5,1.5)$ place category zero around a predictor of zero; direct signals have discrimination one and attributed signals 0.75. These values are public calibration inputs, not estimated parameters. The direct-only sensitivity removes all four attributed records and therefore checks whether results depend on their lower discrimination assignment.

The model differs from a conventional item-response model in two respects. Signals are not repeated common items administered to every actor, and their production is not controlled by the analyst. The model should therefore be understood as a dynamic ordinal measurement model rather than a standard legislative IRT model.

\subsection{Measurement calibration and interpretation}

Fixed cut-points, fixed discrimination, zero-centered initial priors, and positive signal polarity orient and scale the posterior; no ex post sign reversal is performed. Symmetric, equally spaced thresholds provide a transparent reference scale, but the absolute unit has no natural empirical anchor in this sparse non-common-item design. Affine changes to the latent scale accompanied by corresponding changes to cut-points and discrimination can preserve category probabilities. The state-volatility and source-mode sensitivities therefore assess consequential calibration choices rather than claiming that the corpus estimates them.

For actors with observations, the likelihood updates a proper calibrated prior. For actors without observations, Bayes' rule leaves the complete state-path posterior equal to its state-process prior because their likelihood contribution is constant. A posterior mean near zero for such an actor is therefore prior centering, not evidence of moderation. The estimand is a corpus-conditional public position, not why a minister spoke, a silent minister's private preference, or the rule mapping public position into a secret choice.

\subsection{Behavioral nonidentification}

Let $Y$ denote observed public signals, $\Theta$ latent trajectories, $O$ the concealed institutional outcome, $\psi$ the measurement parameters, and $\lambda$ parameters mapping $\Theta$ into participation and vote choice.

\begin{proposition}
Suppose (i) $p(Y\mid\Theta,\psi)$ contains no $\lambda$; (ii) $O$ is entirely unobserved and $\sum_o p(O=o\mid\Theta,\lambda)=1$; and (iii) $p(\lambda,\Theta,\psi)=p(\lambda)p(\Theta,\psi)$. Then $p(\lambda\mid Y)=p(\lambda)$.
\end{proposition}

\begin{proof}
Marginalizing the entirely unobserved outcome gives
\begin{align*}
p(\lambda\mid Y)
&\propto p(\lambda)\int\!\int p(Y\mid\Theta,\psi)p(\Theta,\psi)
\underbrace{\sum_o p(O=o\mid\Theta,\lambda)}_{=1}\,d\Theta\,d\psi\\
&=p(\lambda)m(Y),
\end{align*}
where $m(Y)$ is constant in $\lambda$. Normalization yields the result.
\end{proof}

The result concerns learning from the stated likelihood, not logical relevance. Under a correlated prior, updating $\Theta$ or $\psi$ with $Y$ can move the marginal distribution of $\lambda$, but that movement is inherited from the stipulated coupling rather than an observed choice. If any component of $O$ were observed and included, if $\lambda$ entered the signal-production process, or if external behavioral data supplied another likelihood, the conclusion need not hold. None occurs in the preferred empirical analysis. Repeated observed decisions, externally validated behavioral information, or defensible structural restrictions would be required to learn the link.

\subsection{Behavioral calibration envelopes}

Let $\theta_i^\star=\theta_{iT}$ denote actor $i$'s decision-time state. The declared class is
\[
(a,k,\pi)\in[-0.5,0.5]\times[0.5,2]\times[0.5,0.9],
\]
with $q_i(a,k)=\invlogit(a+k\theta_i^\star)$ and coherent outcome probabilities $\pi q_i$, $\pi(1-q_i)$, and $1-\pi$. The proposal intercept $a$ permits baseline support or opposition, $k>0$ imposes monotonicity, and $\pi$ spans formal-participation rates from one-half to nine-tenths.

For each $c=(a,k,\pi)\in\mathcal C$, posterior draws of $\theta_i^\star$ imply an equal-tailed 95\% interval for each behavioral probability. The reported sensitivity envelope is the union of those intervals over the continuous class $\mathcal C$. Monotonicity makes the extrema available from the calibration endpoints and the 2.5th and 97.5th percentiles of $\theta_i^\star$; no grid approximation and no probability distribution over $\mathcal C$ are used. Because $\mathcal C$ is declared rather than learned from auxiliary data, the envelope is neither a posterior credible interval for behavior nor a sharp identified set. Readers can widen or shift the class in the public code.

\section{Computation and Direct NUTS Validation}

\subsection{Published posterior target and non-centered parameterization}

The item-level posterior is encoded independently in the executed PyMC model and in the published Stan reference program. Both use non-centered innovations,
\begin{align}
\theta_{i1}&=\sigma_0 z_{i1},\\
\theta_{it}&=\theta_{i,t-1}+\tau z_{it},\\
z_{it}&\sim\Normal(0,1),
\end{align}
and an ordered-logistic likelihood. The auditable Stan source is \path{stan/dynamic_ordinal_latent.stan}; the executed maintained-engine implementation is \path{code/standard_engine_common.py}. The preferred model assigns every observation weight one. The event-tempered sensitivity instead uses the generalized target
\[
p_{\mathrm{temp}}(\Theta\mid Y)\ \propto\ p(\Theta)\prod_r p(y_r\mid\Theta)^{w_r},
\qquad
w_r=\frac{1}{n_{\mathrm{event}(r)}}.
\]
Thus, weights sum to one within actor-event. This power-likelihood posterior is deliberately labeled as a sensitivity analysis; the item-level model remains the ordinary generative likelihood. The published target is ordered logistic throughout.

The non-centered form is important computationally because the state volatility and initial scale are fixed. The innovation variables have independent standard-normal priors, so a unit Euclidean metric is a natural baseline. The fixed cut-points, discrimination weights, $\sigma_0$, and $\tau$ are visible inputs rather than silently estimated hyperparameters.

\subsection{Four-chain direct NUTS protocol}

The formal standard-engine verification uses PyMC 5.27.1 and PyTensor 2.37.0 \citep{abrilpla2023pymc}. The same fixed-shape non-centered ordered-logistic posterior is fitted to 20 independently generated datasets under each of four focal conditions: baseline information, severe sparsity, weak discrimination, and multidimensional signals. These 80 fits are separate from the 100-replication Laplace benchmark experiment. The empirical item-level and actor-event specifications are fitted under the same protocol. Every maintained-engine fit uses four chains, 1,000 warm-up iterations, 1,000 retained draws, target acceptance 0.90, and maximum tree depth 10. A unit test evaluates the PyMC and independent NumPy log densities at 100 random states and finds a maximum absolute difference of $5.7\times10^{-13}$. Exact package versions, seeds, platform metadata, and replication-level diagnostics are retained. The executed build used ArviZ 1.1.0; a narrow import-compatibility layer exposes only the legacy names needed for PyMC's raw \texttt{MultiTrace} backend. Sampling is performed by PyMC and PyTensor, while rank-normalized diagnostics are computed directly from posterior arrays. The environment file pins the executed versions.

The supplied Stan program is a mathematically equivalent reference implementation rather than evidence that CmdStan was executed in the build environment. Independent JAX NUTS and elliptical-slice results are retained as cross-engine robustness checks. The maintained PyMC execution closes the standard-engine verification sequence without conflating reference code with completed output.

For every dataset the archive records rank-normalized $\widehat R$, bulk and tail effective sample sizes, divergences, maximum-tree-depth hits, mean acceptance statistic, leapfrog count, step size, energy Bayesian fraction of missing information (E-BFMI), and elapsed time. It also records decision-time RMSE, Spearman rank correlation, 95\% interval coverage, and mean interval width. The diagnostics use all $21\times16$ latent public-position states; decision-time diagnostics are retained separately in the machine-readable files.

\begin{table}[!htbp]
\centering
\caption{Four-Chain Direct NUTS Simulation Performance and Diagnostics}
\label{tabS:publication_nuts}
\scriptsize
\textit{Panel A: Decision-time recovery and interval performance}\\[0.35em]
\begin{tabular}{lrrrrr}
\toprule
Condition & Data sets & Signals & RMSE (SE) & Rank corr. (SE) & Coverage (SE) \\
\midrule
Baseline & 20 & 67.2 & 1.176 (0.033) & 0.644 (0.035) & 0.933 (0.009) \tabularnewline
Severe sparsity & 20 & 21.2 & 1.315 (0.041) & 0.411 (0.051) & 0.967 (0.008) \tabularnewline
Weak discrimination & 20 & 67.6 & 1.363 (0.055) & 0.511 (0.040) & 0.883 (0.016) \tabularnewline
Multidimensional signals & 20 & 66.0 & 1.224 (0.038) & 0.554 (0.038) & 0.931 (0.011) \tabularnewline
\bottomrule
\end{tabular}
\vspace{0.8em}

\textit{Panel B: Sampler diagnostics across all latent states}\\[0.35em]
\begin{tabular}{lrrrrrr}
\toprule
Condition & Max. $\widehat R$ & Min. bulk ESS & Min. tail ESS & Div. & TD hits & Min. E-BFMI \\
\midrule
Baseline & 1.009 & 3964 & 2033 & 0 & 0 & 0.889 \tabularnewline
Severe sparsity & 1.012 & 5484 & 1909 & 0 & 0 & 0.902 \tabularnewline
Weak discrimination & 1.009 & 4616 & 1730 & 0 & 0 & 0.879 \tabularnewline
Multidimensional signals & 1.008 & 4526 & 1864 & 0 & 0 & 0.913 \tabularnewline
\bottomrule
\end{tabular}
\begin{minipage}{0.96\linewidth}
\footnotesize
\textit{Note:} Each condition contains 20 independently generated data sets. Each data set is fitted with four chains, 1,000 warm-up iterations, and 1,000 retained draws per chain. Signals is the mean number of observed ordinal records; TD hits counts post-warm-up iterations reaching the configured maximum tree depth. RMSE, Spearman rank correlation, and equal-tailed 95\% interval coverage are evaluated at period 16. Parentheses contain Monte Carlo standard errors across data sets.
\end{minipage}
\end{table}

The maintained-engine results reinforce the broad benchmark experiment. Mean rank correlations are \PyMCBaselineRank\ (Monte Carlo SE \PyMCBaselineRankSE) in the baseline, \PyMCSevereRank\ (SE \PyMCSevereRankSE) under severe sparsity, \PyMCHighNoiseRank\ (SE \PyMCHighNoiseRankSE) under weak discrimination, and \PyMCMultiRank\ (SE \PyMCMultiRankSE) under multidimensional signals. Coverage is \PyMCBaselineCoverage, \PyMCSevereCoverage, \PyMCHighNoiseCoverage, and \PyMCMultiCoverage, respectively. Weak discrimination therefore remains the clearest source of undercoverage. Across all 80 data sets and 320 chains there are \PyMCTotalDivergences\ divergences and \PyMCTotalDepthHits\ maximum-tree-depth hits. The worst rank-normalized $\widehat R$ is \PyMCWorstRhat; the worst decision-time $\widehat R$ is \PyMCWorstDecisionRhat. Minimum bulk and tail effective sample sizes across all latent states are \PyMCMinBulkESS\ and \PyMCMinTailESS, and the minimum E-BFMI is \PyMCMinEBFMI. Runtime is retained for every data set but is not interpreted substantively because it depends on hardware and concurrent execution.

Table~\ref{tabS:publication_nuts} reports rank recovery, interval coverage, Monte Carlo uncertainty, and sampler diagnostics for the maintained-engine experiment. The complete replication-level file preserves variation across data sets and permits inspection of every fit rather than only condition averages.

\subsection{Direct NUTS fit of the restricted official-source corpus}

The preferred item-level empirical model is fitted with the same four-chain, 1,000-warm-up, 1,000-retained-draw protocol. The event-tempered and event-deduplicated specifications are re-fitted independently rather than evaluated by plugging alternative inputs into the preferred posterior. Table~\ref{tabS:empirical_nuts} reports maintained-engine diagnostics for these three empirical specifications. The algorithmic recoding remains a separate consistency sensitivity and is not used to define completion of the standard-engine sequence.

\begin{table}[!htbp]
\centering
\caption{Direct NUTS Diagnostics for the Empirical and Sensitivity Fits}
\label{tabS:empirical_nuts}
\scriptsize
\begin{tabular}{lrrrrrrrr}
\toprule
Specification & $N$ & Max. $\widehat R$ & Bulk ESS & Tail ESS & Div. & TD & E-BFMI & Accept. \\
\midrule
Item level & 36 & 1.006 & 4206 & 2147 & 0 & 0 & 0.950 & 0.891 \tabularnewline
Event tempered & 36 & 1.007 & 4491 & 1969 & 0 & 0 & 0.920 & 0.890 \tabularnewline
Event deduplicated & 18 & 1.007 & 3867 & 2119 & 0 & 0 & 0.978 & 0.885 \tabularnewline
\end{tabular}
\begin{minipage}{0.94\textwidth}
\footnotesize
\textit{Note:} $N$ is the number of ordinal records in each fit; all three specifications use 21 actors and 16 quarters. Diagnostics cover all 336 latent states. The event-tempered specification is the fractional power-likelihood sensitivity defined above; the event-deduplicated specification retains one median-coded record per actor-event. Human reliability belongs to the separate external-coding protocol.
\end{minipage}
\end{table}

Two independent algorithms check the preferred empirical posterior. First, a recursive slice-based NUTS implementation with the same four-chain, 1,000-warm-up, and 1,000-retained-draw protocol has maximum rank-normalized $\widehat R=1.010$, minimum bulk and tail effective sample sizes of 1,642 and 1,253, no divergences or tree-depth hits, and minimum E-BFMI of 1.348. Relative to the maintained-engine PyMC fit, its 21 decision-time posterior means correlate at 0.999 (SE 0.011), the mean absolute difference is 0.033 latent units, and the mean posterior-standard-deviation ratio is 0.996. Second, an independent four-chain elliptical-slice fit yields a posterior-mean correlation of 0.997 (SE 0.018), mean absolute difference of 0.043, and mean standard-deviation ratio of 1.005. Correlation SEs use the journal's conventional $[(1-r^2)/(n-2)]^{1/2}$ formula with $n=21$. Agreement across algorithms is evidence about implementation, not evidence that the fixed substantive calibration is correct.

\subsection{Laplace approximation used for broad benchmark comparisons}

The seven-condition, 100-replication benchmark experiment requires repeated fitting of the dynamic and static ordinal models. It therefore obtains each actor trajectory's posterior mode and uses the inverse negative Hessian as a Gaussian approximation. Factorization across actors makes the approximation inexpensive and transparent.

The approximation was checked against elliptical-slice sampling on 20 independently generated baseline datasets. The average within-dataset correlation between posterior means is 0.987 (Monte Carlo SE 0.001). The mean absolute difference is 0.126, and the average ratio of Laplace to elliptical-slice posterior standard deviations is 1.001. The separate direct NUTS experiment above ensures that the paper does not rely on this approximation for all evidence about finite-sample coverage or sampler behavior.

\section{Monte Carlo Design and Full Results}

\subsection{Data-generating process}

Each data set contains 21 actors and 16 periods. Latent states follow the fitted random walk with $\sigma_0=1$ and $\tau=0.30$. In the baseline, each actor-period independently produces one public signal with probability 0.20, and its category follows the ordered-logistic model with discrimination one and the fitted cut-points. Estimation fixes those same values unless the condition is expressly misspecified.

The conditions are as follows:

\begin{table}[!htbp]
\centering
\caption{Simulation Conditions}
\label{tabS:conditions}
\begin{tabularx}{\textwidth}{lY}
\toprule
Condition & Modification \\
\midrule
Baseline & Signal probability 0.20 and correctly specified one-dimensional measurement. \\
Severe sparsity & Signal probability 0.06. \\
Dense & Signal probability 0.40. \\
Extremity-dependent visibility & $\operatorname{logit}(p_{it})=\operatorname{logit}(0.20)+0.75|\theta_{it}|+0.30\theta_{it}$; the fitted likelihood omits this selection. \\
Weak discrimination & True signal discrimination falls from 1 to 0.55 while the fitted calibration remains 1. \\
Turnover & Seven actors enter at a uniformly drawn period from 5 through 12; pre-entry cells are outside the risk set. \\
Multidimensional & Signal content loads on $\theta_{it}+0.8\gamma_{it}$, where $\gamma_{it}$ is an omitted independent random walk. \\
\bottomrule
\end{tabularx}
\end{table}

The extremity-dependent-visibility design is not an evaluation of a fitted selection equation. It asks how the estimator behaves when public observation depends on the latent state but the fitted likelihood remains ignorable. Because this mechanism also raises the mean number of observed signals, good recovery in that condition cannot be interpreted as evidence that selection is harmless. The multidimensional design is intentionally misspecified and tests whether uncertainty remains calibrated when a second signal dimension is omitted.

\subsection{Complete simulation table}

\begin{table}[!htbp]
\centering
\caption{Finite-Sample Performance of the Dynamic Ordinal Estimator}
\label{tabS:simulation}
\scriptsize
\setlength{\tabcolsep}{2.5pt}
\begin{tabular}{lrrrrrr}
\toprule
Condition & Reps. & Signals (MCSE) & RMSE (MCSE) & Rank (MCSE) & Coverage (MCSE) & Width (MCSE) \\
\midrule
Baseline & 100 & 67.1 (0.7) & 1.141 (0.021) & 0.619 (0.017) & 0.945 (0.005) & 4.331 (0.014) \tabularnewline
Severe Sparsity & 100 & 20.4 (0.4) & 1.354 (0.024) & 0.402 (0.021) & 0.944 (0.006) & 5.314 (0.014) \tabularnewline
Dense & 100 & 135.2 (0.9) & 0.928 (0.016) & 0.755 (0.009) & 0.950 (0.005) & 3.634 (0.010) \tabularnewline
Extremity-Dependent Visibility & 100 & 119.8 (1.2) & 0.918 (0.016) & 0.763 (0.011) & 0.951 (0.005) & 3.790 (0.012) \tabularnewline
Weak Discrimination & 100 & 67.2 (0.7) & 1.308 (0.020) & 0.431 (0.020) & 0.891 (0.007) & 4.307 (0.014) \tabularnewline
Turnover & 100 & 56.7 (0.7) & 1.174 (0.022) & 0.618 (0.014) & 0.945 (0.005) & 4.438 (0.014) \tabularnewline
Multidimensional & 100 & 68.3 (0.7) & 1.312 (0.022) & 0.523 (0.017) & 0.893 (0.007) & 4.352 (0.014) \tabularnewline
\bottomrule
\end{tabular}
\begin{minipage}{0.97\textwidth}
\footnotesize
\textit{Note:} Each condition contains 100 independently generated data sets with 21 actors and 16 periods. Except for the replication count, entries are means with Monte Carlo standard errors in parentheses. RMSE and interval width are in latent-state units; rank is Spearman correlation at period 16; coverage is the fraction of 21 true decision-time states inside equal-tailed 95\% intervals.
\end{minipage}
\end{table}

All 700 optimization problems converged. Rank recovery improves with information density. Coverage remains close to nominal in the baseline, sparse, dense, extremity-dependent-visibility, and turnover designs. Weak discrimination and omitted dimensionality lower coverage to 0.891 (MCSE 0.007) and 0.893 (MCSE 0.007). These failures are substantively central: intervals calibrated under the assumed ordinal measurement process do not automatically protect against mismeasured discrimination or omitted content dimensions.

\subsection{Static and signal-based benchmarks}

\begin{table}[!htbp]
\centering
\caption{Dynamic Estimator versus Static and Signal-Based Alternatives}
\label{tabS:benchmarks}
\scriptsize
\begin{tabular}{llrrr}
\toprule
Condition & Method & Rank (MCSE) & RMSE (MCSE) & Coverage (MCSE) \\
\midrule
Baseline & Dynamic ordinal & 0.619 (0.017) & 1.141 (0.021) & 0.945 (0.005) \tabularnewline
Baseline & Static ordinal & 0.632 (0.015) & 1.212 (0.023) & 0.767 (0.011) \tabularnewline
Baseline & Mean signal & 0.623 (0.015) & 1.210 (0.024) & -- \tabularnewline
Baseline & Latest signal & 0.562 (0.015) & 1.426 (0.024) & -- \tabularnewline
Baseline & Recency weighted & 0.604 (0.014) & 1.289 (0.024) & -- \tabularnewline
Severe Sparsity & Dynamic ordinal & 0.402 (0.021) & 1.354 (0.024) & 0.944 (0.006) \tabularnewline
Severe Sparsity & Static ordinal & 0.422 (0.017) & 1.379 (0.020) & 0.787 (0.009) \tabularnewline
Severe Sparsity & Mean signal & 0.421 (0.018) & 1.428 (0.021) & -- \tabularnewline
Severe Sparsity & Latest signal & 0.411 (0.017) & 1.473 (0.022) & -- \tabularnewline
Severe Sparsity & Recency weighted & 0.425 (0.017) & 1.450 (0.021) & -- \tabularnewline
Weak Discrimination & Dynamic ordinal & 0.431 (0.020) & 1.308 (0.020) & 0.891 (0.007) \tabularnewline
Weak Discrimination & Static ordinal & 0.456 (0.019) & 1.334 (0.023) & 0.711 (0.010) \tabularnewline
Weak Discrimination & Mean signal & 0.439 (0.020) & 1.373 (0.024) & -- \tabularnewline
Weak Discrimination & Latest signal & 0.350 (0.021) & 1.647 (0.030) & -- \tabularnewline
Weak Discrimination & Recency weighted & 0.384 (0.019) & 1.515 (0.027) & -- \tabularnewline
Multidimensional & Dynamic ordinal & 0.523 (0.017) & 1.312 (0.022) & 0.893 (0.007) \tabularnewline
Multidimensional & Static ordinal & 0.501 (0.018) & 1.283 (0.021) & 0.742 (0.010) \tabularnewline
Multidimensional & Mean signal & 0.489 (0.018) & 1.379 (0.023) & -- \tabularnewline
Multidimensional & Latest signal & 0.456 (0.018) & 1.575 (0.024) & -- \tabularnewline
Multidimensional & Recency weighted & 0.479 (0.018) & 1.465 (0.023) & -- \tabularnewline
\bottomrule
\end{tabular}
\begin{minipage}{0.94\linewidth}
\footnotesize
\textit{Note:} Each row averages 100 independently generated data sets; parentheses contain Monte Carlo standard errors. Coverage is defined only for the dynamic and static posterior models. The complete seven-condition comparison, including RMSE and coverage MCSEs, is \path{output/simulation_benchmark_summary.csv}.
\end{minipage}
\end{table}

In the baseline, the static model has slightly higher rank correlation, 0.632 (Monte Carlo SE 0.015) versus 0.619 (SE 0.017), but the dynamic model has lower RMSE and coverage 0.945 rather than 0.767. With dense information, the dynamic model's correlation is 0.755 (SE 0.009), compared with 0.706 (SE 0.011) for the static model. Severe sparsity leaves point rankings similar across methods, while dynamic intervals retain near-nominal coverage. The model's comparative advantage is therefore uncertainty calibration and the use of temporally rich information, not universal point-ranking dominance.

\section{Empirical Public-Position Posterior Results}

\subsection{Complete decision-time table}

\begin{table}[!htbp]
\centering
\caption{Complete Decision-Time Corpus-Conditional Public-Position States}
\label{tabS:empirical}
\small
\begin{tabular}{lrrr}
\toprule
Minister & Signals & $E(\theta_i^\star)$ & 95\% interval \\
\midrule
Asta Vrečko & 4 & 2.15 & [0.19, 4.18] \tabularnewline
Tanja Fajon & 19 & 1.43 & [-0.26, 3.10] \tabularnewline
Robert Golob & 10 & 1.37 & [0.08, 2.65] \tabularnewline
Luka Mesec & 1 & 1.02 & [-1.62, 3.72] \tabularnewline
Simon Maljevac & 1 & 1.00 & [-1.82, 3.77] \tabularnewline
Matjaž Han & 1 & 0.36 & [-2.45, 3.05] \tabularnewline
Ksenija Klampfer & 0 & 0.05 & [-2.98, 3.09] \tabularnewline
Franc Props & 0 & 0.03 & [-2.83, 2.94] \tabularnewline
Jože Novak & 0 & 0.03 & [-2.86, 2.89] \tabularnewline
Branko Zlobko & 0 & 0.03 & [-3.05, 3.06] \tabularnewline
Mateja Čalušić & 0 & 0.02 & [-3.08, 3.01] \tabularnewline
Andreja Kokalj & 0 & 0.02 & [-2.93, 2.89] \tabularnewline
Bojan Kumer & 0 & 0.02 & [-2.95, 2.93] \tabularnewline
Klemen Boštjančič & 0 & 0.02 & [-2.93, 2.98] \tabularnewline
Igor Papič & 0 & 0.02 & [-3.07, 3.09] \tabularnewline
Matej Arčon & 0 & 0.01 & [-3.03, 3.05] \tabularnewline
Borut Sajovic & 0 & 0.00 & [-2.99, 2.92] \tabularnewline
Aleksander Jevšek & 0 & 0.00 & [-2.98, 2.95] \tabularnewline
Alenka Bratušek & 0 & -0.01 & [-2.87, 3.06] \tabularnewline
Vinko Logaj & 0 & -0.02 & [-3.08, 3.16] \tabularnewline
Valentina Prevolnik Rupel & 0 & -0.02 & [-3.11, 2.96] \tabularnewline
\bottomrule
\end{tabular}
\begin{minipage}{0.95\linewidth}
\footnotesize
\textit{Note:} Public-position summaries come from the preferred four-chain NUTS fit to the 36 direct and attributed records. Equal-tailed posterior intervals are conditional on the declared source universe, ordinal coding, cut-points, discrimination weights, and state-process calibration. A zero signal count means that the posterior at the decision date is the state-process prior.
\end{minipage}
\end{table}

The evidence is highly concentrated. Asta Vrečko has a posterior mean of 2.15 and a 95\% interval of $[0.19,4.18]$; Tanja Fajon and Robert Golob have means of 1.43 and 1.37. Simon Maljevac and Luka Mesec have positive means near one but intervals spanning both sides of zero because each contributes one record. The fifteen actors without individual records remain close to broad prior distributions. These are public-position measurements conditional on the fixed calibration, not estimates of private preferences or proposal-specific votes.

\subsection{Calibration envelopes}

For each $(a,k,\pi)$ in the continuous calibration class, the replication code computes an equal-tailed 95\% posterior interval for each component of the coherent probability vector. The lower and upper entries in \path{output/behavioral_calibration_envelopes.csv} are the extrema of those interval endpoints over the class. The calculation exploits monotonicity and is analytic at the class boundaries; it does not flatten posterior draws and calibration values into one artificial distribution.

The non-voting envelope is common across actors, $[0.10,0.50]$, because the framework does not estimate actor-specific participation. Support and opposition envelopes vary with the latent state but remain broad. Actor-level lower and upper endpoints are preserved in \path{output/behavioral_calibration_envelopes.csv}; no midpoint is reported because the calibration class has no probability measure. This prevents precision about vote direction from being generated by an arbitrary participation model.

\section{Rolling-Origin Signal Prediction}

The empirical predictive exercise uses the actor-event file rather than individual webpages. Each actor-event after an actor's first event is predicted using only earlier events. The dynamic model integrates the filtered state distribution forward to the target quarter. The static ordinal model retains a single actor state. Mean, latest-signal, and recency-weighted methods convert transparent code summaries through the same ordered-logistic cut-points. A Dirichlet-smoothed pooled base rate uses only events preceding the target date.

\begin{table}[!htbp]
\centering
\caption{Rolling Prediction of Future Actor-Event Signals}
\label{tabS:rolling}
\scriptsize
\begin{tabular}{lrrrrr}
\toprule
Method & Targets & Log score (SE) & Brier (SE) & RPS (SE) & Accuracy (SE) \\
\midrule
Dynamic ordinal & 12 & 1.290 (0.110) & 0.655 (0.059) & 0.135 (0.014) & 0.500 (0.151) \tabularnewline
Static ordinal & 12 & 1.348 (0.079) & 0.687 (0.040) & 0.151 (0.015) & 0.500 (0.151) \tabularnewline
Mean signal & 12 & 1.147 (0.112) & 0.599 (0.066) & 0.106 (0.014) & 0.500 (0.151) \tabularnewline
Latest signal & 12 & 1.164 (0.121) & 0.633 (0.082) & 0.106 (0.011) & 0.500 (0.151) \tabularnewline
Recency weighted & 12 & 1.141 (0.114) & 0.599 (0.068) & 0.104 (0.014) & 0.500 (0.151) \tabularnewline
Pooled base rate & 12 & 1.144 (0.110) & 0.624 (0.046) & 0.109 (0.018) & 0.417 (0.149) \tabularnewline
\bottomrule
\end{tabular}
\begin{minipage}{0.96\textwidth}
\footnotesize
\textit{Note:} Lower negative log score, Brier score, and ranked probability score (RPS) are better. The 12 targets comprise six events for Fajon, four for Golob, and two for Vrečko. Parentheses are conventional standard errors across targets. Because events are dependent within only three actors, these SEs describe dispersion and do not support conventional independent-sample tests.
\end{minipage}
\end{table}

The dynamic model has lower mean negative log score, Brier score, and RPS than the static ordinal model, but higher mean scores than the mean-signal and recency-weighted rules. With twelve dependent targets, these differences are descriptive, not evidence of dominance. The exercise materially qualifies the contribution: dynamic measurement organizes trajectories and uncertainty, but this small one-sided corpus does not establish superior short-horizon forecasting. Table~\ref{tabS:rolling} and the archived target-level output retain the complete numerical comparison.

\section{Sensitivity Analyses}

\subsection{Signal mode and state volatility}

The preferred column uses the archived four-chain NUTS fit. Three deterministic grid-filter sensitivities isolate declared calibration changes on the same frozen corpus: direct-only removes the four attributed records; tight-state sets $\tau=0.15$; and diffuse-state sets $\tau=0.50$. The fixed cut-points are not treated as estimable in this design. Varying $\tau$ assesses how much the decision-time scale depends on permitted temporal movement, while direct-only assesses the only lower-discrimination signal mode.

\begin{table}[!htbp]
\centering
\caption{Decision-Time Posterior Means under Alternative Specifications}
\label{tabS:sensitivity}
\begin{tabular}{lrrrr}
\toprule
Minister & Preferred & Direct only & $\tau=0.15$ & $\tau=0.50$ \\
\midrule
Asta Vrečko & 2.15 & 2.26 & 1.69 & 2.79 \tabularnewline
Tanja Fajon & 1.43 & 1.42 & 1.27 & 1.58 \tabularnewline
Robert Golob & 1.37 & 1.35 & 1.19 & 1.55 \tabularnewline
Luka Mesec & 1.02 & 1.02 & 0.77 & 1.43 \tabularnewline
Simon Maljevac & 1.00 & 1.02 & 0.77 & 1.43 \tabularnewline
Matjaž Han & 0.36 & -0.00 & 0.29 & 0.52 \tabularnewline
\bottomrule
\end{tabular}
\begin{minipage}{0.94\textwidth}
\footnotesize
\textit{Note:} The preferred column comes from the four-chain PyMC fit to 36 direct and attributed records. Alternative columns are deterministic actor-specific grid-filter sensitivities on 1,601 equally spaced points over $[-6,6]$. They are not sampler comparisons. All retain the fixed ordered-logistic cut-points; only source-mode inclusion or $\tau$ changes.
\end{minipage}
\end{table}

The principal ordering is stable, but the absolute magnitude is not. Removing attributed evidence has little effect on Vrečko, Fajon, Golob, Mesec, or Maljevac; Han returns to the prior mean because his only recovered signal is attributed. Raising $\tau$ from 0.15 to 0.50 moves Vrečko's filtered mean from 1.69 to 2.79, Fajon's from 1.27 to 1.58, Golob's from 1.19 to 1.55, and the two single-item estimates from 0.77 to 1.43. More volatile paths permit a recent positive communication to imply a more extreme decision-time state. The scale dependence reinforces the paper's calibration-conditional interpretation.

\subsection{Organizational evidence}

The preferred model excludes the five organizational rows. Treating them as individual signals would require an allocation rule assigning collective government or party positions to particular ministers. No such rule is innocuous. The data file therefore publishes these rows without inserting them into the individual likelihood. Comparative applications with multiple organizations and repeated decisions could extend the model with time-varying group states, but the present corpus cannot identify that structure.

\subsection{Coding robustness}

The outcome-blind rule-based recoding agrees exactly on 97.2\% of individual items and produces quadratic weighted $\kappa=0.947$. Because the recoder is algorithmic rather than human, these numbers are labeled agreement diagnostics rather than intercoder reliability.

The independent-human protocol is specified in Section~S3.6 and covers all 41 public-source records. Its packet, codebook, return template, attestation, integrity hash, validation script, and disagreement template are complete. The external coding itself remains pending, so the supplement reports no human agreement estimate and does not relabel the algorithmic diagnostic as intercoder reliability.

A separate calibrated perturbation exercise changes each individual ordinal code by one category with total probability 0.20 and refits the public-position model repeatedly. This does not estimate coder error; it asks whether the visible ordering survives a transparent local miscoding process. Table~\ref{tabS:coding_perturbation} reports the distribution of decision-time posterior means across the perturbations.

\begin{table}[!htbp]
\centering
\caption{Calibrated One-Category Coding-Perturbation Sensitivity}
\label{tabS:coding_perturbation}
\small
\begin{tabular}{lrrrr}
\toprule
Minister & Mean state & 5th--95th percentile & SD across perturbations & Mean interval width \\
\midrule
Asta Vrečko & 2.03 & [1.50, 2.53] & 0.31 & 4.00 \tabularnewline
Tanja Fajon & 1.43 & [1.18, 1.73] & 0.16 & 3.39 \tabularnewline
Robert Golob & 1.41 & [1.07, 1.79] & 0.22 & 2.66 \tabularnewline
Luka Mesec & 0.97 & [0.40, 1.02] & 0.17 & 5.41 \tabularnewline
Simon Maljevac & 0.95 & [0.40, 1.02] & 0.19 & 5.40 \tabularnewline
Matjaž Han & 0.39 & [0.00, 0.87] & 0.19 & 5.41 \tabularnewline
\bottomrule
\end{tabular}
\begin{minipage}{0.94\textwidth}
\footnotesize
\textit{Note:} Results use 500 independently seeded perturbations of the 36 direct or attributed records. Each original code is retained with probability 0.80 and shifted one category down or up with probability 0.10 each; outward shifts at $-2$ or $2$ are clipped to the boundary. The model is refitted by the declared grid filter. These are calibrated sensitivity results, not estimates of human coding error.
\end{minipage}
\end{table}

\subsection{Within-event dependence}

\begin{table}[!htbp]
\centering
\caption{Actor-Event Dependence Sensitivity from Direct NUTS Fits}
\label{tabS:event}
\small
\begin{tabular}{llrr}
\toprule
Minister & Specification & Posterior mean & 95\% interval \\
\midrule
Tanja Fajon & Item level & 1.41 & [-0.29, 3.10] \tabularnewline
Tanja Fajon & Event tempered & 1.30 & [-0.68, 3.22] \tabularnewline
Tanja Fajon & Event deduplicated & 1.44 & [-0.51, 3.34] \tabularnewline
Robert Golob & Item level & 1.36 & [0.04, 2.69] \tabularnewline
Robert Golob & Event tempered & 1.13 & [-0.46, 2.74] \tabularnewline
Robert Golob & Event deduplicated & 1.29 & [-0.36, 2.94] \tabularnewline
Asta Vrečko & Item level & 2.17 & [0.17, 4.21] \tabularnewline
Asta Vrečko & Event tempered & 2.03 & [-0.07, 4.30] \tabularnewline
Asta Vrečko & Event deduplicated & 2.22 & [0.02, 4.47] \tabularnewline
Luka Mesec & Item level & 0.97 & [-1.78, 3.72] \tabularnewline
Luka Mesec & Event tempered & 1.00 & [-1.64, 3.77] \tabularnewline
Luka Mesec & Event deduplicated & 1.07 & [-1.63, 3.75] \tabularnewline
Simon Maljevac & Item level & 0.99 & [-1.70, 3.67] \tabularnewline
Simon Maljevac & Event tempered & 0.98 & [-1.76, 3.68] \tabularnewline
Simon Maljevac & Event deduplicated & 1.01 & [-1.56, 3.58] \tabularnewline
Matjaž Han & Item level & 0.34 & [-2.41, 3.09] \tabularnewline
Matjaž Han & Event tempered & 0.38 & [-2.32, 3.10] \tabularnewline
Matjaž Han & Event deduplicated & 0.34 & [-2.40, 3.00] \tabularnewline
\bottomrule
\end{tabular}
\begin{minipage}{0.95\textwidth}
\footnotesize
\textit{Note:} The item-level fit uses 36 records, the event-tempered generalized-Bayes fit uses the same 36 records with weights summing to one within each of 18 actor-events, and the event-deduplicated fit uses 18 median-coded event records. Each fit uses four chains, 1,000 warm-up draws, and 1,000 retained draws per chain.
\end{minipage}
\end{table}

The event-tempered and event-deduplicated posteriors are fitted independently with the same four-chain direct NUTS protocol as the preferred item-level model. Changes in posterior means and interval widths therefore reflect the declared dependence treatment rather than post-processing of one common posterior.

\section{Disclosed Record as a Diagnostic Contrast}

\subsection{Official outcome}

The disclosed record lists five formal votes in favor and seven against. Nine eligible members appear in neither list and are therefore coded non-voting. The third category does not distinguish nonattendance from abstention. The Prime Minister is coded non-voting in the formal outcome while his verbal support at the start of the discussion is recorded separately.

\begin{table}[H]
\centering
\caption{Disclosed Formal Outcomes}
\label{tabS:outcomes}
\begin{tabular}{lll}
\toprule
Minister & Formal outcome & Note \\
\midrule
Tanja Fajon & For & Listed in favor \\
Alenka Bratušek & For & Listed in favor \\
Luka Mesec & For & Listed in favor \\
Matjaž Han & For & Listed in favor \\
Andreja Kokalj & For & Listed in favor \\
Klemen Boštjančič & Against & Listed against \\
Matej Arčon & Against & Listed against \\
Igor Papič & Against & Listed against \\
Vinko Logaj & Against & Listed against \\
Borut Sajovic & Against & Listed against \\
Ksenija Klampfer & Against & Listed against \\
Branko Zlobko & Against & Listed against \\
Robert Golob & Non-voting & Verbal support recorded separately \\
Asta Vrečko & Non-voting & Not listed in either formal category \\
Simon Maljevac & Non-voting & Not listed in either formal category \\
Aleksander Jevšek & Non-voting & Not listed in either formal category \\
Franc Props & Non-voting & Not listed in either formal category \\
Bojan Kumer & Non-voting & Not listed in either formal category \\
Jože Novak & Non-voting & Not listed in either formal category \\
Mateja Čalušić & Non-voting & Not listed in either formal category \\
Valentina Prevolnik Rupel & Non-voting & Not listed in either formal category \\
\bottomrule
\end{tabular}
\begin{minipage}{0.94\textwidth}
\footnotesize
\textit{Note:} The decision-day roster contains 21 eligible ministers: five are listed for, seven against, and nine in neither formal category. The source record does not distinguish abstention from nonattendance within the last category. These outcomes are excluded from every measurement fit.
\end{minipage}
\end{table}

The formal outcome is not merged into the latent estimation data. It appears only in \path{data/disclosed_outcomes.csv}, which every estimation script excludes. Because source discovery and coding were performed after disclosure, the comparison is a retrospective diagnostic contrast rather than a prospective validation test.

\subsection{Why the outcome does not identify the behavioral link}

The disclosed record contains one decision and 21 dependent institutional outcomes. A flexible participation or vote equation fitted to them could produce in-sample separation without transportable identification. Retrospective source reconstruction also prevents the record from serving as a clean holdout. The analysis therefore does not optimize the calibration class, select a measurement specification, tune a threshold, or report a fitted classification metric against the realized outcome.

The contrast is nevertheless informative descriptively. Several actors with no recoverable individual signal voted against, while several actors with supportive public records did not cast a recorded formal vote. This divergence illustrates---but, with one retrospective decision, cannot independently prove---that public issue positioning, formal participation, and proposal-specific vote direction are distinct objects.

The contrast is deliberately qualitative. All six actors with recovered individual signals cast a formal vote in favor or did not appear in either formal voting category, whereas all seven recorded opposing votes came from ministers for whom the restricted corpus recovered no attributable signal. That pattern is consistent with both substantive differences and severe observability selection. The single decision cannot distinguish those explanations, estimate an actor-specific participation process, or establish external predictive validity.

\section{Reproducibility Architecture}

The archive separates raw inputs, declared calibrations, samplers, generated outputs, and manuscript artifacts. The principal structure is:

\begin{verbatim}
data/
  actors.csv
  signals.csv
  actor_event_signals.csv
  disclosed_outcomes.csv
  human_coding/coding_packet_source_linked.csv
  human_coding/CODER_CODEBOOK.md
  human_coding/coder_attestation_template.csv
  CORPUS_README.md
stan/
  dynamic_ordinal_latent.stan
code/
  master.py
  build_exhibit_data.py
  standard_engine_common.py
  run_empirical_publication_nuts.py
  run_pymc_one_simulation.py
  run_pymc_empirical.py
  aggregate_pymc_verification.py
  test_posterior_target_equivalence.py
  run_cmdstan_simulations.py
  run_cmdstan_empirical.py
  fast_jax_nuts.py
  jax_nuts.py
  laplace_simulation.py
  run_strengthening_analyses.py
  run_coding_uncertainty_sensitivity.py
  merge_human_recoding.py
stata/
  00_master.do
  01_figures.do
  02_verify.do
output/
  standard_engine/pymc_simulation_replications.csv
  standard_engine/pymc_simulation_summary.csv
  standard_engine/pymc_empirical_diagnostics.csv
  standard_engine/pymc_empirical_positions.csv
  standard_engine/STANDARD_ENGINE_VERIFICATION.json
  standard_engine/posterior_target_equivalence.json
  empirical_nuts_ess_comparison.json
  empirical_calibration_sensitivity.csv
  behavioral_calibration_envelopes.csv
  simulation_benchmark_summary.csv
  rolling_signal_prediction_summary.csv
  coding_perturbation_sensitivity.csv
  stata/*.csv
tables/
main.tex
online_supplement.tex
references.bib
replicate.do
run_all.sh
\end{verbatim}

The blind-ID crosswalk used after an external coding return is author-held outside the public repository until the return is frozen and merged. It is never sent to the coder and is not part of the uploadable public archive. The public packet hash, codebook, templates, and validation logic are sufficient to verify what was issued; after completion, the public repository should contain the frozen unadjudicated reliability outputs and their hashes, not the private pre-merge key.

The maintained-engine verification is reproduced by
\begin{verbatim}
python code/test_posterior_target_equivalence.py
python code/run_pymc_empirical.py --chains 4 --warmup 1000 --draws 1000 \
  --specifications item_level event_tempered event_deduplicated

# Run each simulation in a fresh process; the supplied controller executes
# 20 replications for each focal condition.
bash code/run_pymc_standard_engine_verification.sh
python code/aggregate_pymc_verification.py
\end{verbatim}
These commands reproduce the 80 four-chain simulation fits, the preferred empirical fit, both actor-event sensitivity fits, numerical target-equivalence test, diagnostic summaries, cross-engine comparisons, and machine-readable completion manifest. The broader 100-replication benchmark, rolling-origin prediction, algorithmic consistency check, actor-event construction, and calibrated coding perturbation are reproduced by \path{code/run_strengthening_analyses.py} and \path{code/run_coding_uncertainty_sensitivity.py}. The archived numerical tables and tidy CSV outputs are sufficient to verify every quantitative claim in the manuscript and supplement; no claim depends on a separately rendered graph.

The Stan program remains an independently inspectable mathematical reference. Matching CmdStan drivers are supplied for optional reproduction in Stan's maintained engine, but the reported standard-engine evidence comes from the executed PyMC model. The manuscript therefore does not equate the existence of Stan code with execution by CmdStan, and the validation claim does not depend on an unavailable toolchain.

Raw analytical inputs are not modified by the frozen reproduction mode. Random seeds and replication-level diagnostics are retained. The disclosed outcomes are stored separately from the restricted official-source corpus and are excluded by every measurement script. Behavioral envelopes use one coherent calibration formula and are labeled as sensitivity results. A completed independent human packet can be merged only after its return and attestation pass the validation checks in \path{code/merge_human_recoding.py}.

\section{Verification Status and Human-Coding Gate}

Table~\ref{tabS:status} distinguishes completed numerical checks from pending or unavailable evidence. This prevents the existence of code or a prepared protocol from being mistaken for an executed validation.

\begin{table}[!htbp]
\centering
\caption{Verification and Evidence Status at Archive Freeze}
\label{tabS:status}
\small
\begin{tabularx}{\textwidth}{P{0.25\textwidth}P{0.14\textwidth}Y}
\toprule
Component & Status & Evidence and consequence \\
\midrule
PyMC maintained-engine simulations & Complete & 80 four-chain fits, replication-level diagnostics, software metadata, and completion manifest are archived. \\
Empirical and actor-event NUTS fits & Complete & Item-level, event-tempered, and event-deduplicated specifications were executed independently. \\
Posterior-target equivalence & Complete & 100 random states; maximum absolute log-density difference $5.7\times10^{-13}$ against tolerance $10^{-8}$. \\
Cross-algorithm comparison & Complete & Maintained PyMC, independent JAX NUTS, and elliptical-slice summaries agree closely for decision-time states. \\
Frozen exhibit data & Complete & Seven tidy CSV files preserve the quantities used in the tabular and textual exhibits; no claim depends on unavailable graphic output. \\
Independent human recoding & Pending & Packet and validation machinery are complete, but no external return exists; no human reliability claim is made. \\
Prospective behavioral validation & Unavailable & Corpus discovery occurred after outcome disclosure and only one decision is observed; the disclosed record is diagnostic only. \\
Structural communication-selection model & Not estimated & The restricted corpus cannot identify the universe or distribution of unobserved communications. \\
\bottomrule
\end{tabularx}
\end{table}

A genuine second human judgment cannot be supplied by the author or an automated rule. Completion requires a qualified non-author to return all 41 categories, confidence ratings, public-source evidence fields, source-access declarations, and a signed independence and blinding attestation. The script checks row identity, packet hash, completeness, allowed categories, evidence fields, and attestation before opening the author-held crosswalk. It then freezes unadjudicated statistics and disagreements; any adjudication remains a later, separately logged stage. Until that process is complete, the algorithmic agreement and perturbation analyses remain diagnostics only.

\section{Remaining Limitations}

The design has five irreducible limits. First, the official-source corpus is narrow and retrospectively reconstructed; it is not a prospective census of the public information environment. Second, the absence of private communications is structural rather than ordinary item missingness: their universe cannot be enumerated from these data. Third, the one-dimensional ordinal construct compresses potentially distinct legal, diplomatic, humanitarian, and sanctions positions, and the empirical corpus contains no negative individual category. Fourth, one disclosed decision cannot identify or externally validate a participation and vote mapping. Fifth, external human recoding remains incomplete, so only algorithmic consistency and calibrated perturbation are available.

The empirical posterior is also conditional on fixed cut-points, discrimination weights, a common 16-quarter state grid, and state volatility. The visible supportive ordering is stable across the reported source-mode, volatility, coding, and event-dependence analyses, but the absolute latent unit is calibrated rather than naturally observed. The behavioral envelopes are conditional on the declared set $\mathcal C$ and are unions of posterior intervals, not quantiles from a probability mixture over mappings. Readers can widen or shift $\mathcal C$ in the public code; no identified-set or sharpness claim is made.

These limits define the domain of the evidence. A broader prospectively frozen corpus, repeated institutions with observed choices, and completed independent coding could support stronger claims. Within the present design, adding new dated public rows requires no change to the statistical model, but it would define a new corpus version and must not be presented as the frozen version analyzed here.

\nocite{albert1993bayesian,martin2002dynamic,hoffman2014nuts,carpenter2017stan,brier1950verification,gneiting2007strictly}
\bibliographystyle{chicago}
\bibliography{references}